\documentclass[11pt]{article}

\usepackage[english]{babel}
\usepackage[letterpaper,top=2cm,bottom=2cm,left=2cm,right=2cm,marginparwidth=1.75cm]{geometry}
\usepackage{nicefrac}

\usepackage{amsmath}
\usepackage{amssymb}
\usepackage{bbm}
\usepackage{amsthm}
\usepackage{thmtools}
\usepackage{graphicx}
\usepackage{physics}
\usepackage{complexity}
\usepackage[colorlinks=true, allcolors=blue]{hyperref}
\usepackage[capitalise]{cleveref}
\usepackage{url}
\usepackage{bm}
\usepackage{comment}
\usepackage[dvipsnames]{xcolor}
\usepackage{enumitem}

\usepackage[giveninits=true,maxbibnames=99,style=alphabetic,maxalphanames=4,minalphanames=3,isbn=false]{biblatex}

\AtEveryBibitem{%
  \clearlist{language}%
}
\renewbibmacro*{doi+eprint+url}{%
  \iftoggle{bbx:doi}
    {\printfield{doi}}
    {}%
  \newunit\newblock
  \ifboolexpr{togl {bbx:eprint} and test {\iffieldundef{doi}}}
    {\usebibmacro{eprint}}
    {}%
  \newunit\newblock
  \ifboolexpr{togl {bbx:url} and test {\iffieldundef{doi}}  and test {\iffieldundef{eprint}}}
    {\usebibmacro{url+urldate}}
    {}}

\newtheorem{theorem}{Theorem}[section]
\newtheorem*{theorem*}{Theorem} 

\newtheorem{claim}[theorem]{Claim}
\newtheorem{lem}[theorem]{Lemma}
\newtheorem{fact}[theorem]{Fact}

\newtheorem{proposition}[theorem]{Proposition}

\theoremstyle{definition}
\newtheorem{definition}[theorem]{Definition}

\crefname{lem}{Lemma}{Lemmas}
\crefname{claim}{Claim}{Claims}

\newcommand{\eps}{\varepsilon}

\DeclareMathOperator*{\Ex}{\mathbb{E}}

\newclass{\QCSZK}{QCSZK}

\usepackage{authblk}

\title{Quantum state isomorphism problems for groups}

\author{}
\author[1]{Alexandru Gheorghiu}
\author[2]{Dale Jacobs}
\author[2]{Saeed Mehraban}
\author[3]{Arsalan Motamedi}

\affil[1]{\small IBM Research}
\affil[2]{Tufts University}
\affil[3]{University of Waterloo}

\date{}

\begin{document}
\maketitle

\sloppy

\begin{abstract}

We study the computational complexity of quantum state isomorphism problems under group actions: given two quantum circuits that prepare pure or mixed states, decide whether the two states are related by a group action. This can be seen as a quantum state version of the \emph{Hidden Shift Problem}, in much the same way that the \emph{State Hidden Subgroup Problem}, introduced in~\cite{bouland2025state}, is a quantum version of the ordinary Hidden Subgroup Problem. 

We prove several results for this computational problem:
\begin{itemize}
    \item For the pure-state version, we show that the problem is $\BQP$-hard for all nontrivial groups, and contained in $\QCMA \cap \QCSZK$. When the quantum state is given as a polynomial stabilizer-rank state, we obtain an upper bound of $\mathsf{NP} \cap \SZK$. We further obtain refined results for specific groups of interest: for abelian groups we show that the problem reduces to the state hidden subgroup problem over the generalized dihedral group; for the Clifford group, the problem is at least as hard as \emph{Graph Isomorphism} under polynomial-time reductions; for the Pauli group it is $\BQP$-complete.

    \item For the mixed-state version, for nontrivial, finite and efficiently representable groups, the problem is $\QSZK$-complete.

    \item We also study a variant of this problem over an infinite group, in particular, the bosonic linear optical unitaries. We show that in the setting where the classical description of the quantum state is given in a suitable wave function representation known as the stellar representation, the problem is at least as hard as Graph Isomorphism, and is contained in $\NP \cap \SZK$.
\end{itemize}
Prior to our work, state isomorphism problems had only been studied for the symmetric group~\cite{lockhart2017quantum}. As a consequence of our results, we resolve an open question posed in \cite{hinsche2025abelian} about the existence of a quantum algorithm for the abelian state hidden subgroup problem on mixed states. We show that this problem is $\QSZK$-hard in the worst case, thereby ruling out an efficient quantum algorithm unless $\QSZK = \BQP$.

\end{abstract}

\newpage

\tableofcontents

\newpage

\section{Introduction}

Isomorphism problems, such as group and graph isomorphism, are central to classifying mathematical structures up to symmetry. In classical complexity theory, they form an important class of $\NP$-intermediate problems which have significant applications such as pattern/shape matching, network comparison, or symmetry classification, and admit statistical zero-knowledge proofs but seem to have no classical or quantum polynomial-time algorithms. 

Recently, \cite{lockhart2017quantum} initiated the study of a quantum state isomorphism problem for the symmetric group: given two quantum states, decide whether they are equivalent up to a permutation of the physical locations of the qubits. They showed that detecting isomorphism under qubit permutations is at least as hard as Graph Isomorphism, and lies within the class of problems admitting a quantum statistical zero-knowledge protocol ($\QSZK$), for pure states. Several questions, however, remained open, including whether the quantum communication in the $\QSZK$ protocol can be made classical, and establishing upper bounds for the mixed-state version of the problem.

In this work, we generalize this line of research to state isomorphism problems over \emph{arbitrary groups}. We show that mixed-state isomorphism over finite groups is $\QSZK$-complete, thereby establishing it as a natural complete problem for this class alongside canonical problems such as quantum state distinguishability~\cite{watrous2002limits}. We further eliminate the aforementioned need for quantum communication in the $\QSZK$ protocol for pure states from~\cite{lockhart2017quantum}, thereby showing that the pure-state version of the problem is contained in $\QCSZK$ for all nontrivial, efficiently representable groups. For the same types of groups, we also show that the problem is $\BQP$-hard. 
We further investigate specific group families: for the Clifford group, we prove hardness via a reduction from the Graph Isomorphism problem, while for the Pauli group, we show that the problem is $\BQP$-complete, thereby capturing the full power of polynomial-time quantum computation. 

Our results also establish a connection between state isomorphism problems and the recent $\textsc{StateHSP}$ framework introduced in \cite{bouland2025state}, which studies the task of identifying a hidden subgroup that stabilizes a given quantum state. As a consequence, we address an open question of \cite{hinsche2025abelian} by showing that the $\textsc{StateHSP}$ framework is unlikely to extend to mixed states. Specifically, we prove that the mixed $\textsc{StateHSP}$ problem is $\QSZK$-hard, thereby ruling out efficient quantum algorithms for the mixed-state version of the problem, assuming $\QSZK \neq \BQP$.
Furthermore, we observe that the state isomorphism problem for groups is in fact the quantum analog of the \emph{Hidden Shift Problem}~\cite{friedl2014hidden, van2006quantum,childs2005quantum}. In other words, just as \textsc{StateHSP} is a quantum state version of HSP, in which the task is to find a hidden subgroup that stabilizes a quantum state, we find that state isomorphism is a quantum state version of the Hidden Shift Problem~\cite{friedl2014hidden}, in which one wishes to find a ``shift,'' given by the action of a group element, that maps one quantum state to another. Using this observation, we show that the pure state isomorphism problem for abelian groups reduces to \textsc{StateHSP} for the generalized dihedral group.

Finally, we initiate the study of a bosonic variant of the problem, considering the isomorphism of infinite-dimensional quantum states under Gaussian transformations, corresponding to an infinite non-compact group (see \cite{weedbrook2012gaussian,ferraro2005gaussian} for a review of Gaussian quantum information). We show that for succinctly represented families of bosonic states, this problem is at least as hard as Graph Isomorphism. This complements the results of \cite{migdal2014multiphoton, parellada2023no}. In addition, we obtain an $\mathsf{NP}$ and an $\mathsf{SZK}$ upper bound.

\subsection{Model and summary of the results}

We define various state group isomorphism problems as variations of the following problem:

\begin{definition}[Pure State Group Isomorphism Problem; {$(\alpha,\beta)$-\textsc{PSGI}$[G]$}]
    For a finite group $G$ with an efficient unitary or projective unitary representation $R : G \to U(2^n),$ the input consists of two quantum circuits $C_1$ and $C_2$ specifying states $\ket{\psi_1} = C_1 \ket{0^n}$ and $\ket{\psi_2} = C_2 \ket{0^n}.$ The goal is to decide which of the following statements is true (promised that one of them is true):
\begin{itemize}
    \item There exists $g \in G$ such that $\abs{\bra{\psi_1} R(g)\ket{\psi_2}} \geq \beta$,
    \item For all $g \in G$, $|\bra{\psi_1} R(g)\ket{\psi_2}| \leq \alpha$.
\end{itemize}
\label{def:PSGI}
\end{definition}

Essentially, the problem asks whether the output state of one circuit (when acting on $\ket{0^n}$) can be mapped to the output state of the other circuit with an element from the group $G,$ or whether the states are far from each other under the action of any group element. Alternatively, we can view the task as testing whether a quantum state is close to or far from the \emph{orbit} of another quantum state under the action of the group $G$. Note that in this definition, quantum states are succinctly represented using the circuits that prepare them. We also consider a variant of this definition in which states are specified by a polynomial rank stabilizer decomposition (see \cref{sec:low-stab-rank-isomorphism} for a concrete definition). It should also be noted that the group is not provided as input but is part of the problem description. In other words, the problem is defined relative to a specific group. Throughout the paper, we will consider specific versions of this problem for groups of interest, such as the Pauli and Clifford groups, respectively.

As previously mentioned, this problem can be seen as a quantum state analog of the Hidden Shift Problem~\cite{friedl2014hidden, van2006quantum, childs2005quantum}. In the decision version of that problem, one is given access to two functions $f : \mathbb{Z}_N \to S$ and $g : \mathbb{Z}_N \to S$ and promised that either there exists a shift $s \in \mathbb{Z}_N$ such that $f(x) = g(x+s)$ for all $x \in \mathbb{Z}_N,$ or the functions are far from being related by a hidden shift.\footnote{Far here means that the Hamming distance between the truth tables of $f(x)$ and $g(x+s)$ is large for all non-zero $s \in \mathbb{Z}_N.$}
We can see how the state isomorphism problem is a natural extension of this. In particular, there is a simple reduction from the Hidden Shift Problem to \textsc{PSGI}$[\mathbb{Z}_N].$ Given $f$ and $g,$ we construct the states
\[
\ket{\psi_f} = \frac{1}{\sqrt{N}} \sum_{x \in \mathbb{Z}_N} \ket{x} \ket{f(x)} \quad\quad \ket{\psi_g} = \frac{1}{\sqrt{N}} \sum_{x \in \mathbb{Z}_N} \ket{x} \ket{g(x)}. 
\]
We observe that if $f$ and $g$ are related by the hidden shift, $s,$ then
\[
\ket{\psi_f} = \frac{1}{\sqrt{N}} \sum_{x \in \mathbb{Z}_N} \ket{x} \ket{g(x+s)},
\]
and so $R(s) \ket{\psi_f} = \ket{\psi_g},$ where $R(s)$ is the unitary representation of the cyclic shift by $s$ on the first register, i.e. $R(s) \ket{x} \ket{y} = \ket{x + s}\ket{y}.$\footnote{All additions are assumed to be modulo $N$.}
Conversely, if $f$ and $g$ are far from being related by a hidden shift, the corresponding states $\ket{\psi_f}$ and $\ket{\psi_g}$ will be far from each other under the action of any element in $\mathbb{Z}_N.$
\textsc{PSGI} is, of course, a more general problem, as it allows for arbitrary (efficiently preparable) states and arbitrary (efficiently representable) group actions.

We also note that \textsc{PSGI} is trivially contained in $\mathsf{QCMA},$ since in the case where the states are related by a group element, $g,$ the witness is simply a classical description of $g$. A quantum verifier could then test the overlap condition by preparing the associated states, applying $R(g)$ to one of them, and then performing a SWAP test.

Our first result is that 
\begin{theorem} [The Complexity of PSGI] The following upper and lower bounds hold for \textsc{PSGI}$[G]$:

     \begin{itemize}
         \item \textbf{Upper bounds:} Let $G$ be any efficiently represented finite group. Then $(\alpha,\beta)$-\textsc{PSGI}$[G]$ is contained in $\QCSZK$ for some constant $\alpha$ and $\beta = 1 - \dfrac{1}{\poly(n)}$. Moreover, for abelian $G$, $(\alpha,\beta)$-\textsc{PSGI}$[G]$ reduces to an approximate version of \textsc{StateHSP}$[G \ltimes \mathbb{Z}_2]$ (see \cref{subsec:psgi-abelian-hsp}).

         \item \textbf{Lower bounds:} For any efficiently represented nontrivial finite group $G$, $(\alpha,\beta)$-\textsc{PSGI}$[G]$ is \textsf{BQP}-hard, even when $\alpha = o(1)$ and $\beta = 1 - \exp(-n)$. When $G$ is chosen to be the $n$-qubit Clifford group, $\mathcal{C}_n$, $(\alpha,\beta)$-\textsc{PSGI}$[\mathcal{C}_n]$ is at least as hard as graph isomorphism ($\textsc{GI}$) for constant $\alpha$ and  $\beta$. The $\textsc{GI}$-hardness remains even if the quantum states are succinctly represented as linear combinations of stabilizer states (see \cref{sec:low-stab-rank-isomorphism} for more details).
     \end{itemize} 
     \label{thm:PSIG}
\end{theorem}
\noindent 
The $\QCSZK$ upper bound statement is restated and proved in \cref{subsec:qcszk}, while the reduction to \textsc{StateHSP} for abelian groups is shown in \cref{subsec:psgi-abelian-hsp}. For restatements and proofs of the lower bounds, see \cref{sec:PSGI-BQP-hard} \cref{subsec:cip-gi-hardness} and \cref{sec:low-stab-rank-isomorphism}.

Previously, lower and upper bounds for \textsc{PSGI} had only been considered for the symmetric group. Specifically, it was shown in~\cite{lockhart2017quantum} that \textsc{PSGI}$[S_n]$ is $\textsc{GI}$-hard and contained in $\QSZK \cap \mathsf{QCMA}$.

Next, we fully classify the complexity of \textsc{PSGI}$[G]$ for the Pauli group $G = \mathcal{P}_n$ by proving that it is equivalent to $\BQP$. 

\begin{theorem} [The complexity of pure state Pauli isomorphism]
\label{thm:main-pauli-bqp-completenes}
\textsc{PSGI}$[\mathcal{P}_n]$ is $\BQP$-complete. 
\end{theorem}
\noindent This result is  proved in \cref{subsec:pauli-bqp-completeness}.

From the perspective of viewing \textsc{PSGI} as a generalization of the Hidden Shift Problem, this is perhaps at first surprising, as there's no known version of that problem that's $\BQP$-complete. In fact, there is also no known version of the Hidden Subgroup Problem that is $\BQP$-complete. On the other hand, it is perhaps less surprising, given that the input to the problem consists of arbitrary quantum circuits that can encode $\BQP$-hard instances. Indeed, we note that the problem is $\BQP$-hard for all non-trivial efficiently representable groups. However, the $\BQP$ upper bound in this case follows directly from the properties of the Pauli group and does not hold, as far as we can tell, for other groups of interest, such as Clifford, permutation, or even cyclic groups.
The fact that the states are pure is also an essential feature that influences the problem's complexity.

Next, we are naturally motivated to consider a version of the state isomorphism problem for mixed states. Let $\mathcal{D} \left( \mathbb{C}^{2^n} \right)$ be the space of $n$-qubit density matrices. In a slight notational abuse, we will denote the \emph{square root fidelity} between two mixed states $\rho$ and $\sigma$ as $F (\rho, \sigma) = \Tr (|\sqrt{\rho} \sqrt{\sigma}|)$. We use the square-root fidelity because it is more convenient for our proofs. Our results also hold for standard fidelity with the appropriate conversions.

\begin{definition}[$(\alpha,\beta)$-Mixed State Group Isomorphism Problem; {$(\alpha,\beta)$-\textsc{MSGI}$[G]$}]
    Let $G$ be a finite group with an efficient unitary representation $R : G \rightarrow U(2^n)$, and let $0 \leq \alpha < \beta \leq 1$.
    Given the description of two quantum circuits $C_1$ and $C_2$ outputting mixed states $\rho_1$ and $\rho_2 \in \mathcal{D} \left( \mathbb{C}^{2^n} \right)$, respectively, the problem is to decide which of the following statements holds (promised that one of them does):
\begin{itemize}
    \item There exists $g \in G$ such that $F(\rho_1, R(g) \rho_2 R(g)^\dagger) \geq \beta$.
    \item For all $g \in G$, it is the case that $F(\rho_1, R(g) \rho_2 R(g)^\dagger) \leq \alpha$
\end{itemize}
\end{definition}

\noindent We show 
\begin{theorem} [Informal; The complexity of mixed state group isomorphism] 
\label{thm:mixed-qszk-complete}
Let $G$ be any nontrivial efficiently representable finite group, and let $R$ be a unitary representation of $G$ such that $\max_{g \neq I}\dfrac{|\Tr(R(g))|}{2^n} = \mu < 1$. Then $(\alpha,\beta)-$MSGI$[G]$ is \textsf{QSZK}-complete under randomized reductions for any $\alpha > \dfrac{1}{2} + \dfrac{\mu}{2}$ and $\beta > \alpha.$
\end{theorem}
\noindent For a proof of this statement, see \cref{sec:mixedstates}.
This result shows that for mixed states, and many choices of $\alpha$ and $\beta$, the group structure does not affect the problem's complexity---it remains $\QSZK$-complete for all nontrivial groups! 

Let us also comment briefly on the trace condition in this statement. The trace condition has a natural interpretation. For a unitary \(U\in U(d)\),
we have \(|\operatorname{Tr}(U)|\le d\), with equality if and only if
\(U=e^{i\theta}I\) is a global phase times the identity. Therefore,
\[
    \max_{g\neq I}\frac{|\operatorname{Tr}(R(g))|}{d}<1
\]
is equivalent to requiring that no nonidentity element of \(G\) is represented
by a global phase. Equivalently, the induced homomorphism
\[
    G\to PU(d):=U(d)/U(1)
\]
is faithful. This is the relevant notion for mixed states, since mixed states transform by
conjugation:
\[
    \rho\mapsto R(g)\rho R(g)^\dagger.
\]
If \(R(g)=e^{i\theta}I\), then this action is trivial on every density matrix.
Thus such a group element is invisible to the mixed-state isomorphism problem.
The trace condition rules out this degeneracy. Moreover, it is always possible to find representations that satisfy this condition. For instance, the left-regular representation has $\mu = 0,$ since every nonidentity group element is represented as a traceless permutation matrix.

\iffalse
We introduce the group isomorphism problem for quantum states, in particular for the Clifford/Gaussian group. 
\fi

\paragraph{Extending to infinite-dimensional settings:}
The previous results were concerned with finite-dimensional systems and finite groups. We also initiate the study of the problem for infinite-dimensional systems and groups, specifically \emph{continuous variable} bosonic systems with Gaussian unitaries as the group.

Bosonic quantum states live in an infinite-dimensional Hilbert space, corresponding to quantum states of $n$ harmonic oscillator modes. This space is spanned by the Fock basis $\{\ket{n_1,\ldots,n_m} : n_i \in \mathbb{N}\}$, where $n_i$ denotes the number of excitations (photons) in mode $i$. Due to the infinite number of degrees of freedom, it is useful to adopt a canonical representation that captures physically relevant states. One such representation is the \emph{stellar representation}~\cite{chabaud2020stellar, chabaud2022holomorphic}, in which a quantum state is described by a holomorphic function that is square-integrable with respect to a Gaussian measure. More specifically, the multimode stellar representation of a bosonic state $\ket{\psi} = \sum_{j_1, \ldots, j_n \geq 0} \psi_{j_1, \ldots, j_n} \ket{j_1, \ldots, j_n}$ is the holomorphic function $F_\psi : \mathbb{C}^m \rightarrow \mathbb{C}$ defined by:
\[
F_\psi(z_1,\ldots,z_n) := \sum_{j_1, \ldots, j_n} \frac{\psi_{j_1, \ldots, j_n}}{\sqrt{j_1! \ldots j_n!}}z_1^{j_1} \ldots z_n^{j_n}.
\] 
The function $F_\psi$ is entire and square-integrable with respect to the Gaussian measure
\[
\frac{1}{\pi^m}\int_{z \in  \mathbb{C}^m} |F_\psi(z)|^2 e^{-\|z\|^2} \, d^{2m} z = 1,
\]
which, via the Hadamard-Weierstrass factorization theorem, ensures that any such state can be written as
\[
F_\psi (z_1, \ldots, z_n) = P(z_1,\ldots,z_n)\, G(z_1,\ldots,z_n),
\]
where $P : \mathbb{C}^n \to \mathbb{C}$ is a polynomial and $G : \mathbb{C}^n \to \mathbb{C}$ is a Gaussian function. The degree of $P$ is called the \emph{stellar rank} of the state. When $G \equiv 1$, then the quantum state is called a core state.

A core bosonic state with rank $r$ has a finite $r$-particle cutoff (i.e., it is supported on sectors containing at most $r$ particles). The corresponding truncated Hilbert space over $n$ modes has dimension
\[
    \sum_{k=0}^r {n+k-1\choose k}={n+r\choose r}=n^{O(r)}.
\]
For $r=O(1)$, these states are efficiently representable.

Given the stellar representation of two quantum states, it is known that if they are within the same Gaussian orbit of each other, then they have the same stellar rank. We can furthermore ask ``given two quantum states with the same stellar rank, are they necessarily within a Gaussian orbit of each other?'' If we compute the core state for both states and we find them to be the same, we are certain that the states are within the Gaussian orbits of each other. If, however, they correspond to different-looking core states with the same degrees, then they may still be related to each other via a linear optical transformation. This motivates us to define the isomorphism problem for core states under linear optical transformations. 

We use $\mathcal{O}_n$ to denote the set of passive linear optical unitaries on $n$ bosonic modes, i.e., unitaries generated by Hamiltonians quadratic in the canonical operators that preserve total photon number. Equivalently, $U \in \mathcal{O}_n$ if its action on annihilation operators is linear:
\[
U^\dagger a_i U = \sum_{j=1}^n V_{ij} a_j, \qquad V \in \mathrm{U}(n).
\]
Here $a_i$ and $a_i^\dagger$ denote the annihilation and creation operators for mode $i$, which respectively remove and add a single excitation and satisfy the canonical commutation relations $[a_i,a_j^\dagger]=\delta_{ij}$. Equivalently, the action of a linear optical network is specified by a unitary element $V \in \mathrm{U} (n)$ which maps the stellar function $F (z_1, \ldots, z_n)$ to $F ((Vz)_1, \ldots, (Vz)_n)$ where $(V z)_j = \sum_{k = 1}^n V_{jk} z_k$.

This defines a homomorphism $R$ from $\mathrm{U}(n)$ to $\mathcal{O}_n$, and the induced action on Fock space gives a unitary representation of $\mathrm{U}(n)$ that is unique up to a global phase. In this sense, $\mathcal{O}_n$ is projectively isomorphic to $\mathrm{U}(n)$.

\begin{definition}
[$(\alpha,\beta; r)$-Pure State Linear Optical Isomorphism Problem]
Given two core states $\ket{c_1}, \ket{c_2}$ over $n$ modes with $r \geq 0$ photons, represented as their core state polynomials, and precision parameters $(\alpha,\beta)$, then $(\alpha,\beta; r)$-PSGI$^*$[$\mathcal O_n$] is the problem of deciding which one of the following two statements hold:
\begin{itemize}
    \item There exists a linear optical unitary $U\in\mathrm{U}(n)$ such that $\abs{\bra{c_2}R(U)\ket{c_1}} \geq \beta$.
    \item For all $U\in\mathrm U (n)$ we have $\abs{\bra{c_2}R(U)\ket{c_1}} \le \alpha$.
\end{itemize}
Here, the superscript $*$ is used to denote a succinct representation of quantum states using core state polynomials.
\end{definition}

 We are interested in understanding how the computational complexity of the above problem depends on $r$. \cite{migdal2014multiphoton} showed that the isomorphism problem with $r = 2$ photons is efficiently solvable. We show that the problem over three photons becomes at least as hard as Graph Isomorphism. We further show that the general problem is contained in $\SZK \cap \NP$ for certain parameter choices. 

\begin{theorem}[Informal; The complexity of linear optical isomorphism]
For any $\beta > \alpha \geq 1-\frac{1}{96n^5}$ we have that $(\alpha,\beta; 3)$-$\mathrm{PSGI}^*[\mathcal O_n]$ is hard for Graph Isomorphism, and for any $\beta-\alpha\ge1/2^{\mathsf{poly}(n)}$ the problem is contained in $\mathsf{NP}$. In addition, for any constant $\alpha<1$ and $\beta = 1-O(\frac{1}{n^{r}})$ the problem $(\alpha,\beta; r)$-$\mathrm{PSGI}^*[\mathcal O_n]$ is contained in $\mathsf{SZK}$.
\label{thm:bosonic-PSGI*}
\end{theorem}

\subsection{Towards applications}
\label{subsc:applications}

There are many situations in physics and computer science in which the physically or operationally relevant objects are not the quantum states themselves, but \emph{equivalence classes} of states under a symmetry group.\footnote{Indeed, quantum states themselves are defined as equivalence classes of complex vectors.} The \emph{state isomorphism problem under a group action}, which we also view as \emph{a quantum state hidden shift problem}, provides a natural decision problem for determining whether two states represent the same information modulo such symmetries. Based on this insight, we highlight several instances in which the state isomorphism problem can be identified and discuss potential applications.

\paragraph{Logical information and Pauli frames in quantum error-correcting codes}

In quantum error correction, physical states encode logical information only up to the action of operators that act trivially on the code space. A quantum code is degenerate if distinct physical errors act identically on the encoded logical information. For stabilizer codes, degeneracy is equivalent to the existence of stabilizer Pauli elements with weight smaller than that of the code distance. 
Given two physical states and a description of the stabilizer group, one may ask whether they correspond to the same logical state. We can view this as an instance of the state isomorphism problem under Pauli group actions.

Moreover, in many fault-tolerant architectures Pauli corrections are not physically applied but instead tracked classically as a so-called Pauli frame. The physical state of the system is then only defined up to an unknown Pauli operator. One could envision applications of Pauli group isomorphism problem to certifying the output of fault-tolerant gadgets modulo propagated Pauli corrections, or, comparing experimental states to ideal targets up to tracked corrections.

\paragraph{Testing quantum states up to symmetry transformations}

While in quantum state tomography one aims to reconstruct a quantum state from measurement data, in many scenarios the relevant object is a state up to a symmetry transformation, such as a local basis change or gauge invariance. 

An example is that of determining whether two states exhibit the same entanglement structure. In other words, we say that two quantum states belong to the same entanglement class if they are related by local basis changes. This naturally leads to a formulation of the group state isomorphism problem for the group $\mathrm{SU}(2)^{\otimes n}$. Entanglement equivalence can also be defined with respect to LOCC operations, leading to a more complex group structure.
Another example is determining whether two quantum states possess the same non-stabilizer (“magic”) structure, which can be viewed as an instance of Clifford state isomorphism.

Condensed matter physics also has a natural equivalence class of states, namely \emph{topological phase}. Two states are said to be in the same topological phase if there exists a simple transformation mapping one to the other, such as a low-depth quantum circuit~\cite{chen2010local}. One could also try to relate this to the group state isomorphism problem, however low-depth quantum circuits do not form a group. We however note that the group structure is not essential in several of the lower bounds established in this paper. In those cases, we merely require that the valid isomorphisms do not form a $t$-design (see \cref{sec:PSGI-BQP-hard}). This is somewhat reminiscent of the result in \cite{schuster2025random} which characterizes the complexity of detecting topological order. Their proof relies on the set of operations forming a $t$-design, which is similar to our use of $t$-designs. See also item 5 in open questions \cref{sec:open-questions}. 
One can also consider a restricted version of topological phase in which the low-depth circuits do form a group, in which case it should be possible to assess its complexity through the state isomorphism framework.

\paragraph{Cryptographic transformations as group actions}

Many quantum cryptographic primitives can be viewed as applying a group action to a quantum state. For example, the quantum one-time pad applies a random Pauli operator to encrypt a state.

Given a plaintext state and a purported ciphertext, one may ask whether the ciphertext lies in the orbit of the plaintext under the allowed key group. This reduces to a state isomorphism problem. More generally, state isomorphism captures whether two states are related by a valid secret-key transformation, providing a natural abstraction for reasoning about the correctness of cryptographic transformations.

\paragraph{Dynamical orbits under Hamiltonian evolution}

Unitary time evolution induces a group action on quantum states. Given a Hamiltonian $H$, one may consider the orbit $\{ e^{-iHt} \ket{\psi} : t \in \mathbb{R} \}$ of an initial state. The problem of determining whether two states lie in the same orbit is a state isomorphism problem under this dynamical action.

When the evolution has finite order, this reduces to a hidden-shift-type problem over a finite group. More generally, it addresses whether two states are related by time evolution under a fixed Hamiltonian. This connects state isomorphism to problems in quantum algorithms, dynamical complexity, and the study of long-time behavior of quantum systems.

\subsection{Discussion, related work and open questions}
\label{sec:open-questions}
In this work, we initiate the study of quantum state isomorphism problems for general groups. As mentioned, this can be viewed as an extension of~\cite{lockhart2017quantum}, which considered isomorphism problems only with respect to a specific representation of the permutation group. It can also be viewed as studying quantum analogs of well-known classical isomorphism problems, such as graph isomorphism, graph automorphism, group isomorphism, and group automorphism, all of which are \NP-intermediate.
From this perspective, classical group isomorphism corresponds to group-state isomorphism, which can be viewed as a state version of the hidden shift problem~\cite{van2006quantum}. Similarly, the classical group automorphism problem can be mapped to the group-state automorphism problem and can be interpreted as a state-level analog of the hidden subgroup problem~\cite{bouland2025state}. Given the variations in complexity of these problems, depending on whether states are pure or mixed, on how they are represented, and on the properties of the underlying groups, developing a better understanding of these problems is an important goal for future research. 

\paragraph{Related work}
The paper \cite{lockhart2017quantum} introduces the state isomorphism problem for the symmetric group, proves $\QSZK$-hardness for the mixed state version, and proves Graph Isomorphism hardness for the pure state version.
The paper \cite{bouland2025state} introduces the state hidden subgroup problem (StateHSP) and \cite{bouland2025state, hinsche2025abelian}, give efficient quantum algorithms for StateHSP over abelian groups. 
The paper \cite{rethinasamy2024quantum} studies a related problem of testing whether states are close or far from being invariant under \textit{every} element of a group $G$, and gives $\QSZK$-completeness results for this problem.

We end this section with a number of open questions:
\begin{enumerate}
    \item The recent work of~\cite{bouland2025state}, which introduced the quantum state hidden subgroup problem, used it to design quantum protocols for tasks related to understanding the entanglement structure of quantum states, such as identifying a hidden cut in a pure state. A natural open question arising from our work is finding applications of quantum algorithms for state isomorphism problems. For example, we wonder whether the techniques developed here can be used to characterize or learn the entanglement structure of quantum states, analogous to \cite{bouland2025state}.

    \item Another question is whether we can provide a more complete characterization of the pure-state group isomorphism problem. We've shown that this problem lies between $\BQP$ and $\QCSZK \cap \mathsf{QCMA}$, and that, for certain groups, such as the Clifford group, it is also at least as hard as graph isomorphism. It is desirable to find tighter lower and upper bounds. Can we improve the upper bound to $\BQP^{\GI}$ for specific groups like the Clifford group?

    \item For mixed state isomorphism, can we improve the tolerance to error parameters, particularly completeness, by improving from $\beta = 1- 1/\mathrm{poly}$ to a constant? This would be desirable because many natural $\QSZK$-complete problems remain $\QSZK$-complete even with constant completeness and soundness parameters, and so one might expect the same to be true for MSGI.

    \item Many of the results in this work address finite groups. It is natural to study more general groups with infinitely many elements. The bosonic isomorphism problem takes a step in this direction by establishing lower and upper bounds on Gaussian isomorphism problems between states. However, many other tools need to be developed to address these problems for bosonic Hilbert spaces and infinite groups in general. 

    \item The isomorphism problem we studied is for ``groups.'' One may wonder how much group structure is actually needed for the main results of this paper. In particular, would we retrieve similar complexity results if we define isomorphism problems with respect to a set of operations? The $\BQP$ lower bounds are likely independent of the group structure, since the $\BQP$-hardness is encoded in the circuits given as input. However most of the upper bounds do make use of the group structure. This is true for the $\BQP$ containment for the Pauli group as well as the $\QCSZK$ and $\QSZK$ upper bounds for pure and mixed states, respectively. The $\QCMA$ upper bound, on the other hand, does not require group structure, merely that the isomorphism can be represented and applied efficiently. It would be interesting to understand to what extent group structure affects the complexity landscape of these isomorphism problems.

    \item Can we use the hardness of group state isomorphism problems discussed in this paper to establish no-go theorems for learning families of quantum states? For instance, ~\cite {iosue2025higher} shows that Gaussian orbits of single-photon sources are efficiently learnable. They give a characterization in terms of families of states that are fully specified by their (up to) fourth moments. Can we use the lower bounds developed in this work to establish no-go theorems for learning states that require greater than $4$ moments to be fully specified? More generally, what do isomorphism problems tell us about the hardness of learning states that arise in the context of certain groups?

    \item \cite{migdal2014multiphoton} argues that to test if two states are related by a Gaussian, it is enough to check several invariants. We know that GI can be embedded as an instance of a Gaussian group state isomorphism. Is it possible to use this to come up with an efficient quantum algorithm for the graph isomorphism problem?

    \item Can we solve linear optical equivalence between the outputs of two CV circuits in $\QCSZK$? In \cite{chabaud2025energy} it was shown that CV circuits with polynomial energy (and under a mild condition about the gate set) can be simulated in $\BQP$. Then, it would be plausible to assume a $\QCSZK$ upper bound here. However, using the simulation technique of \cite{chabaud2025energy} we need a cutoff $M=O(1/\varepsilon^2)$ to achieve precision $\varepsilon$, and following the proof of \cref{thm:qcszk}, we realize that the prover needs an $\varepsilon'$-net for the linear optical gates. As the $\varepsilon'$-net for energy-constrained diamond-norm of energy $M$ requires $N=e^{O(M/\varepsilon')}=e^{O(1/(\varepsilon^2\, \varepsilon'))}$ many elements, in YES cases, we need $(1-O(\varepsilon))^{O (\log N)} = o(1)$, which is impossible to satisfy given the scaling of $N$. We ask whether this limitation is fundamental, i.e., the bosonic problem is harder, or it can be circumvented by a different proof technique.
\end{enumerate}

\subsection*{Acknowledgements}
SM and DJ are grateful to the National Science Foundation (NSF CCF-2013062) for supporting this project. We are also grateful to several people for insightful discussions. AG thanks Adam Bouland and William Kretschmer for early discussions on the complexity of state isomorphism problems and Zhi Li for useful comments regarding fidelity bounds. DJ thanks Jackson Morris and Lorenzo Leone for helpful discussions. SM and AM thank Ulysse Chabaud and Michael Joseph for insightful discussions on bosonic isomorphism problems. The authors are also grateful to the Simons Institute, where some of this work took place.

\subsection*{AI Statement}
We used Claude Opus 4.6 to aid us in the proof of~\Cref{claim:twirls}. Specifically, Claude made us aware of the Rotfel'd inequality from~\cite{rotfel1969singular} which was instrumental in proving the claim. We have also used ChatGPT and Google Gemini for general checking, simplifying and rewriting of certain proofs and specifically to aid in deriving the proof of \cref{lem:aux}. ChatGPT was also used to brainstorm ideas in \cref{subsc:applications}. All claims, proofs, and references involving AI have been verified by the authors.

\section{Technical overview}
In this section, we outline the proofs of the main results.

\subsection*{Pure State Group Isomorphism is in $\QCSZK$}

We'll start with containment in $\QCSZK$ for the pure state version of the state isomorphism problem.\footnote{As mentioned in the previous section, containment in \textsf{QCMA} is immediate by observing that the classical witness is a circuit representation of the group element which maps between the two states.} For detailed proofs, see \cref{subsec:qcszk}.

Let us first informally define $\QCSZK$. The complexity class $\QCSZK$ (quantum classical statistically zero knowledge) is the set of decision problems solvable by an interactive protocol between a quantum polynomial-time verifier and a computationally unbounded prover such that
\begin{itemize}
    \item The communication between the prover and the verifier is entirely classical,
    \item The verifier learns only whether the answer to the decision problem is YES or NO, without learning anything more about the solution. This is the zero-knowledge condition.
\end{itemize}
We now give an overview of the upper bound in \cref{thm:PSIG}. To build some intuition, we recap the proof from \cite{lockhart2017quantum} showing the containment of state isomorphism with respect to the permutation group (\textsc{PSGI}$[S_n]$) in $\QSZK$, which is the same as $\QCSZK$, except the communication between the prover and the verifier is allowed to be quantum.
It is easier to explain this containment for the \emph{complement} of the problem, as both $\QSZK$ and $\QCSZK$ are closed under complement. In other words, given as input circuits $C_1$ and $C_2$ and denoting $\ket{\psi_1} = C_1 \ket{0^n},$ $\ket{\psi_2} = C_2 \ket{0^n},$ one should answer YES if $R(g)\ket{\psi_1}$ is far from the ray (i.e., an element of the projective Hilbert space) of $\ket{\psi_2}$ for all $g \in G$, and answer NO if there exist $g \in G$ and a phase $e^{i\theta}$ such that $\ket{\psi_2}=e^{i\theta}R(g)\ket{\psi_1}$.

The protocol works as follows. The verifier will pick $j \in \{1, 2\}$ uniformly at random, as well as a random group element $g \in G$. It will then prepare $N = \mathrm{poly}(n)$ copies of the state $R(g) \ket{\psi_j}$ and send them to the prover. The prover must decide whether the copies correspond to $\ket{\psi_1}$ or $\ket{\psi_2}$, i.e., to output the value of $j$ that the verifier used. If the states do not lie in the same orbit with respect to the group $G,$ then the copies of $R(g) \ket{\psi_j}$ will be distinguishable and the prover will answer correctly with an overwhelmingly high probability. On the other hand, if the states can be mapped to one another under the action of some group element, the prover will be unable to correctly determine $j$ with probability greater than $1/2$. Moreover, in the YES case, the verifier learns only that the states are distinguishable (and therefore not related by a group element), and nothing else; hence, the protocol is zero-knowledge.

To prove containment in $\QCSZK$, the only difference with respect to the above protocol is that the verifier will send a \emph{classical shadow} of $R(g) \ket{\psi_j},$ instead of sending the state itself~\cite{huang2020predicting}. It is known that for all pure states that are preparable by polynomial-size quantum circuits, there exists an efficient procedure to construct the classical shadow~\cite{aaronson2018shadow, morimae2022one}. One such procedure is for the verifier to sample random Clifford circuits, apply them to the state $R(g) \ket{\psi_j}$, and then measure the output. With overwhelming probability, these measurement outcomes will be unique to the state $R(g) \ket{\psi_j}$ and therefore the prover is able to identify this state.\footnote{Note that the procedure for identifying the state will, in general, not be efficient and require exponential time. However, this is not an issue since the prover is assumed to be computationally unbounded.}
The protocol then proceeds exactly as in the $\QSZK$ case, with the prover correctly outputting $j$ in the YES case and failing to do so with probability greater than $1/2$ in the NO case. 
The main reason the communication can be made classical is the existence of an efficient classical shadow procedure for pure states. For mixed states, however, the number of shadows required for fidelity estimation grows with the effective rank, and so this procedure only works for states that are pure or close to pure.

\subsection*{Pure State Group Isomorphism is $\BQP$-hard}

Next, we prove that \textsc{PSGI}$[G]$ is \textsf{BQP}-hard for all nontrivial groups. For details, see \cref{sec:PSGI-BQP-hard}. Consider a quantum circuit\footnote{We are using a simplified presentation here. The more precise statement is to consider a uniformly generated family of quantum circuits for each input size.} $C_x$ which solves an arbitrary \textsf{BQP} decision problem on input $x$. Without loss of generality, we can assume $C_x$ is in such a way that $C_x\ket{0^n} \approx_{\epsilon} \ket{0^n}$ if $x$ is a YES instance and $C_x\ket{0^n} \approx_{\epsilon} \ket{1^n}$ if $x$ is a NO instance. Through standard amplification arguments, we can take $\epsilon = 2^{-O(n)}$, so that the circuit output in the two cases is negligibly close to either $\ket{0^n}$ or $\ket{1^n}$, respectively. In fact, without loss of generality we can take $C_x$ to perform the mappings $C_x\ket{0^n} \approx_{\epsilon} \ket{0^n}$ if $x$ is a YES instance and $C_x\ket{0^n} \approx_{\epsilon} \ket{\phi}$ if $x$ is a NO instance, where $\ket{\phi}$ is any efficiently preparable state with negligible fidelity to $\ket{0^n}$. Now, taking $\ket{\phi}$ to be some state which is far from $R(g)\ket{0^n}$ for all $g$, we have that if $x$ is a YES instance then $I\ket{0^n}$ and $C_x\ket{0^n}$ are isomorphic, and if $x$ is a NO instance then they are far from isomorphic. To complete the proof, we need to specify how $\ket{\phi}$ is constructed explicitly. For certain groups, such as Paulis or Cliffords, we can simply take $\ket{\phi} = \ket{T^n},$ with $\ket{T} = \frac{1}{\sqrt{2}}(\ket{0} + e^{i \pi/4} \ket{1}),$ as no Pauli or Clifford operation can map $\ket{0^n}$ to an $n$-fold product of magic states (even approximately). In general, however, we can take $\ket{\phi}$ to be a $t$-design state, where $t = \omega(\log |G|).$ For instance, we can construct $\ket{\phi}$ as the output of a random quantum circuit of depth $\textrm{poly}(\log |G|).$ It can then be shown that, with overwhelming probability, no group element in $G$ (under the unitary representation) can map $\ket{0^n}$ to $\ket{\phi}.$ The one caveat of this result is that it makes the reduction randomized, rather than deterministic. It is unclear whether, for every finite group, there exists an efficient deterministic procedure for selecting an appropriate $\ket{\phi}$ state. We leave the resolution of this as an open question.

\subsection*{Pure State Pauli Isomorphism is in $\BQP$}

Next, we examine the complexity of \textsc{PSGI}$[G]$ for specific groups of interest. 

Despite the lower bound of $\BQP$ for all nontrivial groups, we find that the problem can exhibit different complexities depending on the group of interest. Specifically, for the Pauli group, we show that the problem is contained in $\BQP$, and therefore is $\BQP$-complete (\cref{thm:main-pauli-bqp-completenes}; see \cref{subsec:pauli-bqp-completeness} for more details).

Given descriptions of circuits $C_1$ and $C_2$ preparing $\ket{\psi_1}$ and $\ket{\psi_2}$, define the phase-invariant states
\[
    \ket{\psi_i^{\star}}=\ket{\psi_i}\otimes \ket{\overline{\psi_i}},
    \qquad i\in\{1,2\}.
\]
Note that $\ket{\psi_i^\star}$ is efficiently preparable since given a classical description of $C_i$, one can construct a classical description of $\overline{C_i}$ by taking the entrywise conjugate of each gate in $C_i.$ Applying $\overline{C_i}$ to $\ket{0^n}$ then yields $\ket{\overline{\psi_i}}.$ Thus, applying $C_i$ and $\overline{C_i}$ in parallel prepares the state $\ket{\psi_i^\star}.$

For each Pauli $P$, we also define $P^\star=P\otimes\overline P$. If the original instance has perfect completeness, then
\[
    P\ket{\psi_1}=e^{i\theta}\ket{\psi_2}
\]
for some phase $e^{i\theta}$, while
\[
    P^\star\ket{\psi_1^\star}=\ket{\psi_2^\star},
\]
i.e. the complex phases between the first and second registers cancel. Moreover, every $P^\star$ is an involution, and all such operators commute: any commutation phase acquired by the first tensor factor is cancelled by its conjugate in the second factor. Thus
\[
    \mathcal P_n^\star:=\{P\otimes\overline P:P\in\mathcal P_n\}
    \cong\mathbb Z_2^{2n}
\]
is abelian (where Paulis differing only by a scalar phase define the same element).

Now consider the efficiently preparable state
\[
    \ket{\Psi}
    =\frac{1}{\sqrt2}
    \left(\ket0\ket{\psi_1^\star}+\ket1\ket{\psi_2^\star}\right).
\]
It is stabilized by $X\otimes P^\star$ in the perfect-completeness case. More generally,
\[
    \bra{\Psi}(X\otimes P^\star)\ket{\Psi}
    =\abs{\bra{\psi_1}P\ket{\psi_2}}^2.
\]
Taking polynomially many copies amplifies the soundness while preserving inverse-polynomially close completeness. The resulting \textsc{StateHSP} instance is over the elementary abelian group $\mathcal P_n^\star\times\mathbb Z_2$, so weak Fourier sampling and Gaussian elimination recover an odd approximate stabilizer. A final overlap estimate distinguishes the two promises. This proves the $\BQP$ containment of \textsc{PSGI}$[\mathcal P_n]$; the quantitative analysis appears in \cref{subsec:pauli-bqp-completeness}.

As mentioned, the proof crucially relies on the fact that we can ``abelianize'' the Pauli group so as to reduce to a version of \textsc{StateHSP} that is efficiently solvable quantumly.
We remark that this is yet another example of the relationship between state isomorphism problems and the state hidden shift problem. In particular, the $\BQP$ containment of \textsc{PSGI}$[\mathcal{P}_n]$ is related to the fact that the hidden shift problem over $\mathbb{Z}_2^n$ reduces to the hidden subgroup problem over $\mathbb{Z}_2^{n+1}$.\footnote{As a consequence of this observation, any group which admits an efficient algorithm for the hidden shift problem also has an efficient algorithm for the related isomorphism problem.}

For other groups, however, it is not always clear that one can perform this ``abelianization'' trick. Indeed, when considering the Clifford group, this seems unlikely to be the case because it contains the permutation group as a subgroup. This means that one would have to solve an instance of HSP or \textsc{StateHSP} with respect to the permutation group and this is not known to be efficient. This doesn't immediately rule out an efficient quantum algorithm for \textsc{PSGI}$[\mathcal{C}_n]$, as there could be some other approach towards solving the problem that doesn't reduce to an instance of \textsc{StateHSP}. We next argue that this is unlikely to be the case, unless there is a polynomial-time quantum algorithm for the Graph Isomorphism problem.

\subsection*{Reduction from Graph Isomorphism to Pure State Clifford Isomorphism}

Our next result shows that Graph Isomorphism for a graph with $n$ vertices reduces to \textsc{PSGI}$[\mathcal{C}_{n + 1}]$ (\cref{thm:main-pauli-bqp-completenes}; see \cref{subsec:cip-gi-hardness} for more details). The proof is similar to the proof that Graph Isomorphism reduces to \textsc{PSGI}$[\mathcal{S}_n]$ ($\mathcal{S}_n$ is the symmetric group) given in \cite{lockhart2017quantum}. Their proof follows from observing that two graphs $G_1$ and $G_2$ are isomorphic if and only if their graph states $\ket{G_1}$ and $\ket{G_2}$ are isomorphic up to permuting qubits. Since graph states are stabilizer states, this approach doesn't immediately apply to the Clifford group, since the two states are always related by a Clifford operator, regardless of whether the graphs are isomorphic. To remedy this, we instead construct the states
\begin{align}
\begin{split}
    \ket{\psi_1} &= \frac{1}{\sqrt2}\left( \ket{R}\ket{R^n} + \ket{R_-}\ket{G_1} \right) \\
    \ket{\psi_2} &= \frac{1}{\sqrt2}\left( \ket{R}\ket{R^n} + \ket{R_-}\ket{G_2} \right),
\end{split}
\end{align}
where $\ket{R} = \frac{1}{\sqrt2}(\ket{0} + e^{\pi i /8}\ket{1})$ and $\ket{R_-} = \frac{1}{\sqrt2}(\ket{0} - e^{\pi i /8}\ket{1})$.
Essentially, the $\ket{R^n}$ component enforces that any Clifford unitary approximately mapping $\ket{\psi_1}$ to $\ket{\psi_2}$ must be a permutation up to global phase, thereby reducing the Clifford isomorphism problem to the isomorphism problem over the symmetric group. The reason for using $\ket{R^n}$ over the seemingly more natural choice $\ket{T^n}$, for the magic state $\ket{T} = \frac{1}{\sqrt{2}} (\ket{0} + e^{i \pi/4} \ket{1})$, is because $\ket{T}$ is a Clifford eigenstate and this would allow non-permutation Cliffords to map between the states. Note also that we used $\ket{R}$ and $\ket{R_-}$ for controlled qubits. The reason for not using $\ket{0}$ and $\ket{1}$ is to avoid pathological instances of controlled operations that leave the graph state invariant and act as identity on the $R^n$ state, on the $\ket{1}$ branch.

In more detail, the choice of states $\ket{\psi_1}$ and $\ket{\psi_2}$ allows us to perform the reduction as follows. When the graphs are isomorphic, there exists a qubit permutation unitary (which implements the graph isomorphism) that maps $\ket{\psi_1}$ to $\ket{\psi_2}.$ This permutation has the effect of mapping $\ket{G_1}$ to $\ket{G_2}$ while leaving the symmetric $\ket{R^n}$ component unchanged. When the graphs are not isomorphic, we must show that no Clifford operation can map one to the other, even approximately. Suppose, for the sake of contradiction, that such a Clifford existed. We denote it as $D \in \mathcal{C}_{n+1},$ and have that
\[
| \bra{\psi_2} D \ket{\psi_1} | > 0.99999.
\]
Expanding out the left-hand side, we have
\[
|\bra{\psi_1} D \ket{\psi_2}| = \frac{1}{2} | \bra{R_-}\bra{G_1} D \ket{R_-} \ket{G_2} + \bra{R_-}\bra{G_1} D \ket{R} \ket{R^n} + \bra{R}\bra{R^n} D \ket{R_-} \ket{G_2} + \bra{R}\bra{R^n} D \ket{R} \ket{R^n} |.
\]
Now note that for any Clifford, $D$, the middle two terms will be negligible, since no Clifford can map a stabilizer state to an $n$-fold product of magic states with better than $2^{-O(n)}$ fidelity.
This means that in order to achieve $| \bra{\psi_2} D \ket{\psi_1} | > 0.99999,$ it must be that $|\bra{R_-}\bra{G_1} D \ket{R_-} \ket{G_2}|\approx 1$ and $|\bra{R}\bra{R^n} D \ket{R} \ket{R^n}| \approx 1.$ We show in~\Cref{lem:perm} that $|\bra{R}\bra{R^n} D \ket{R} \ket{R^n}| \approx 1$ implies that $D$ must act as a permutation, up to global phase, on the register with the $n$ magic states. However, since $G_1$ and $G_2$ are not isomorphic, such a permutation cannot make $|\bra{R_-}\bra{G_1} D \ket{R_-} \ket{G_2}|\approx 1.$ In other words, it is not possible for a Clifford $D$ to simultaneously map one graph state to the other, with high fidelity, and approximately stabilize a product of magic states. This provides the desired contradiction, concluding the proof and showing that the ability to decide \textsc{PSGI}$[\mathcal{C}_n]$ allows one to solve Graph Isomorphism.

\subsection*{Mixed State Group Isomorphism is $\QSZK$-complete}

Next, we prove that the mixed state isomorphism problem is $\QSZK$-complete for any nontrivial finite group $G$ (\cref{thm:mixed-qszk-complete}; see \cref{sec:mixedstates} for more details). 

We start with containment in $\QSZK$. The $\QSZK$ protocol for mixed state group isomorphism is almost exactly the same as the pure state protocol, where the verifier chooses either $\psi_1$ or $\psi_2$ (which this time are mixed states), applies a random group element, and sends copies of the result to the prover. To prove that this is indeed a valid protocol, we must show that the resulting mixtures over states are indistinguishable if $\psi_1$ and $\psi_2$ are isomorphic under the group action, and distinguishable if they are not. 
The former follows from noting that if there is an $h \in G$ such that $R(h) \psi_1 R(h)^\dagger = \psi_2,$ then
\[
\frac{1}{\abs{G}} \sum_{g \in G} (R(g) \psi_2 R(g)^\dagger) = \frac{1}{\abs{G}} \sum_{g \in G} (R(g) R(h) \psi_1 R(h)^\dagger R(g)^\dagger) = \frac{1}{\abs{G}} \sum_{g \in G} (R(g h) \psi_1 R(h g)\dagger).  
\]
As the sum ranges over all elements in $G$, this is identical to
\[
\frac{1}{\abs{G}} \sum_{g \in G} (R(g) \psi_1 R(g)^\dagger).
\]
Hence, the mixtures of $\psi_1$ and $\psi_2$ averaged over all group elements will be statistically identical.

We now show that if no group element relates $\psi_1$ and $\psi_2,$ then the states the prover receives will be distinguishable. Specifically, we will assume that for all $g \in G,$ it is the case that $F(\psi_1, R(g) \psi_2 R(g)^\dagger) \leq \alpha,$ for some constant $\alpha > 0.$ The goal is then to prove that the resulting $k$-twirled states will have large trace distance for sufficiently large $k$. 
By $k$-twirled, we mean the following states
    \[
        \psi_1' = \frac{1}{\abs{G}}\sum_{g \in G} (R(g) \psi_1 R(g)^\dagger)^{\otimes k}, \quad\quad
        \psi_2' = \frac{1}{\abs{G}}\sum_{g \in G} (R(g) \psi_2 R(g)^\dagger)^{\otimes k}.
    \]
By making use of a trace inequality due to Rotfel'd \cite{rotfel1969singular}, we are able to show that if $F(\psi_1, R(g) \psi_2 R(g)^\dagger) \leq \alpha,$ for all $g \in G,$ then $F(\psi_1, \frac{1}{\abs{G}} \sum_{g \in G} R(g) \psi_2 R(g)^\dagger) \leq \alpha|G|.$ In some sense, this is like a union bound for fidelities. The formal statement is \Cref{claim:twirls}. Applying this to the $k$-twirled states and using the fact that $F(\rho^{\otimes k}, \sigma^{\otimes k}) = F^k(\rho, \sigma),$ we find that
\[
F(\psi_1', \psi_2') \leq \alpha^k |G|.
\]
Since $\alpha$ is a constant, it suffices to take $k = \omega(\log|G|)$ in order to have that $F(\psi_1', \psi_2') = \mathrm{negl}(n),$ where $\mathrm{negl}(n)$ denotes a negligible function in $n$\footnote{We say that a function $f(n)$ is negligible, if for any polynomial $p(n),$ it is the case that $\lim_{n \to \infty} f(n) p(n) = 0$.}.
This means that the states are distinguishable, and so the prover will be able to correctly reply with whether the states it received resulted from copies of $\psi_1$ or $\psi_2$, respectively. As in the pure-state case, the protocol is zero-knowledge, and thus we have containment in \textsf{QSZK}.

To show $\QSZK$-hardness, we reduce from the canonical $\QSZK$-complete problem Quantum State Distinguishability (QSD) \cite{watrous2002quantum}, which asks whether two mixed states are close or far in trace distance. The idea is similar to the reduction from Graph Isomorphism to Pure-State Clifford Isomorphism, in which we use a linear combination of states to enforce that the isomorphism must take a particular form. In particular, we enforce that the only valid isomorphism is the identity. Concretely, if $\sigma_1$ and $\sigma_2$ are the inputs to the QSD instance, we take the states
\begin{align*}
    \rho_1 &= \frac{1}{2}\sigma_1 + \frac{1}{2}\psi \\
    \rho_2 &= \frac{1}{2}\sigma_2 + \frac{1}{2}\psi,
\end{align*}
to be the inputs to the isomorphism problem, where $\psi$ is far from the any state in the orbit of $\sigma_1$ or $\sigma_2$, and $\psi$ is far from invariant under any non-identity group action. Thus, if $\sigma_1$ and $\sigma_2$ are close in trace distance, $\rho_1$ and $\rho_2$ will remain close under the identity element. But if $\sigma_1,\sigma_2$ are far, there is no group element that will simultaneously map $\sigma_1$ to $\sigma_2$ and $\psi$ to $\psi$. Furthermore, there is no group element that will map either $\sigma_1$ or $\sigma_2$ to $\psi$. Thus, $\rho_1$ and $\rho_2$ are not isomorphic. To obtain such a $\psi$, we can simply use the output of an approximate state $t$-design, which will satisfy the requirement for the reduction with high probability. 
Also note that, given circuits that prepare $\sigma_1$ and $\sigma_2$, we can construct efficient circuits that prepare $\rho_1$ and $\rho_2$, respectively.

As with the $\BQP$-hardness of \textsc{PSGI} proof, we remark that this particular reduction is randomized, but for a fixed group we can make it deterministic by replacing $\psi$ with two fixed pure states $\psi_1, \psi_2$ which are far from invariant under any non-identity group action, and occupy distinct orbits. If $\sigma_1,\sigma_2$ are mixed, these states will never have high overlap with either $\sigma_1$ or $\sigma_2$, and if $\sigma_1,\sigma_2$ are pure, they cannot share the same orbit with both $\psi_1$ and $\psi_2$. However, choosing such states $\psi_1,\psi_2$ deterministically will depend on the group. For example, for the Clifford group, we can choose two magic states which occupy different orbits, such as $\ketbra{R^n}{R^n}$ and $\ketbra{R^{n/2}\rangle|0^{n/2}}{R^{n/2}|\langle 0^{n/2}}$, where $\ket{R} = \frac{1}{\sqrt{2}}(\ket{0} + e^{i\pi/8} \ket{1})$ is a non-stabilizer state (which is also not a Clifford eigenstate).

\subsection*{Mixed StateHSP is $\QSZK$-hard for some abelian groups}

As a corollary of our proof that mixed state group isomorphism is $\QSZK$-hard, we obtain that a mixed state version of the StateHSP problem from~\cite{bouland2025state} is $\QSZK$-hard for any abelian group that contains an involution, as well as for the Pauli group. Since many abelian groups do contain an involution, this rules out the possibility of extending the abelian StateHSP framework to mixed states in general, unless $\QSZK = \BQP$. This resolves an open question from~\cite{hinsche2025abelian}, though as this is a worst-case result, it leaves open the possibility of an efficient algorithm for specific instances.

The reduction from the aforementioned Quantum State Distinguishability problem to MixedStateHSP is as follows, where $\sigma_1,\sigma_2$ are the inputs to the QSD instance.
Let $h \in G$ be a group element of order $2$, and let 
$\ket{v_1}, \ket{v_2}$ be orthonormal vectors such that $R(h)\ket{v_1} = \ket{v_2}$ and $R(h)\ket{v_2} = \ket{v_1}$. We call these the \textit{label} vectors.
We take as input to the mixed StateHSP instance the state
\begin{align*}
    \rho = \frac12\ketbra{v_1}{v_1} \otimes \sigma_1 + \frac12\ketbra{v_2}{v_2} \otimes \sigma_2,
\end{align*}
which is efficiently preparable.
Now, we define the representation $R'$ as
\begin{align*}
    R'(g) = R(g) \otimes I,
\end{align*}
where $R(g)$ acts on the label register.
Note that if $\sigma_1,\sigma_2$ are close in trace distance, then $\rho$ and $R'(h)\rho R'(h)^\dagger$ will be close in trace distance as well. But if $\sigma_1,\sigma_2$ are far in trace distance, then $\rho$ and $R'(h)\rho R'(h)^\dagger$ will also be far in trace distance. Thus, if MixedStateHSP returns a hidden subgroup $H$, it suffices to check whether $h \in H$. This is exactly the subgroup membership problem, which is known to be efficient for finite abelian groups \cite{mckenzie1987parallel}. Thus, if we can solve MixedStateHSP, we can check whether $\sigma_1$ and $\sigma_2$ are close or far in trace distance, thereby solving a \textsf{QSZK}-hard problem. 
See \cref{subsec:mixedHSP-qszk-hard} for more details.

\subsection*{Reduction from Graph Isomorphism to Linear Optical Isomorphism and the $\mathsf{SZK}$ upper bound}

Transitioning to bosonic systems and following a method similar to that used for the reduction of Graph Isomorphism (GI) to Pure State Clifford Isomorphism, we prove a GI-hardness result for the Pure State Linear Optical Isomorphism problem (\cref{thm:bosonic-PSGI*}; see \cref{sec:bosons} for more details). Let $A^{(1)}$ and $A^{(2)}$ represent the adjacency matrices of two graphs $G_1, G_2$. We consider the following two states
\begin{align*}
\ket{\psi_1} &\propto \sum_{i} \ket{3_i} + \sum_{ij} A^{(1)}_{ij} \ket{1_i 1_j},\\
\ket{\psi_2} &\propto \sum_{i} \ket{3_i} + \sum_{ij} A^{(2)}_{ij} \ket{1_i 1_j},
\end{align*}
where 
\begin{align*}
\ket{3_i}:=|0,0,\cdots,\underbrace{3}_{i\text{-th location}}, \cdots, 0\rangle, \quad \text{and}\quad  \ket{1_i1_j}:= |0,\cdots 0,\underbrace{1}_{i\text{-th location}},0,\cdots,0,\underbrace{1}_{j\text{-th location}},0 \cdots, 0\rangle.
\end{align*}
In other words, the state $\ket{3_i}$ corresponds to having three photons in mode $i$ (and no photons elsewhere), while $\ket{1_i1_j}$ corresponds to having one photon in mode $i$ and one photon in mode $j$ (with no other photons elsewhere).
Note that if we were to use $\sum_{ij} A_{ij}\ket{1_i1_j}$ alone, the two states would have been isomorphic if and only if the spectra of the two adjacency matrices were the same, and this is indeed not equivalent to GI. However, the addition of $\sum_{i}\ket{3_i}$ resolves this issue. The reason is tied to the fundamental fact that for any $p\neq 1,2$, the only linear transformations that preserve $\ell_p$ norm are permutation matrices (up to a diagonal phase operator). We use this fact for $p=3$, together with a natural connection between bosonic states and complex polynomial functions, to conclude that any linear optical map that transforms $\sum_i \ket{3_i}$ to itself, must be a permutation (up to a diagonal phase operator). This would enforce the two states $\ket{\psi_1}, \ket{\psi_2}$ to be isomorphic if and only if the two adjacency matrices can be mapped into one another via conjugation by a permutation, i.e., the GI problem.

To show that the problem remains hard for GI even in the approximate settings, we show that if there is a linear optical unitary that makes the overlap of the two states larger than $1-1/(96n^5)$, then there exists a permutation that makes the overlap equal to $1$. In order to show this, we first prove that any linear map that approximately preserves $\ell_3$ norm is almost a permutation (up to a diagonal phase matrix). We then use the fact that if the overlap between two adjacency matrices exceeds a threshold, they must be equal. We obtain soundness by contraposition of the previous statement.

Finally, in order to give an $\mathsf{SZK}$ upper bound, we reduce the problem to the Statistical Difference problem in the following way: since the Hilbert space with a constant photon cutoff has dimension $n^{O(r)}$, we can send classical descriptions of these states. For input states $\ket{\psi_1}, \ket{\psi_2}$, the verifier picks one state uniformly at random, applies a random linear optical unitary and an independent uniformly random global phase, and adds a small amount of noise. The global-phase randomization makes the sampler depend only on the state's ray. If the two states are approximately related, up to phase, by a linear optical element, then the resulting random vectors have small total variation distance, and vice versa.

\section{Preliminaries and Notation}

\subsection{Basic definitions and tools}
We use $\Omega(\cdot), O(\cdot), \omega(\cdot), o(\cdot)$ to hide universal constants.
We follow standard notation for quantum states, unitaries, channels, measurement operators, and Hilbert spaces as in~\cite{nielsen2010quantum}.
We use $D(\rho, \sigma) = 1/2\Tr |\rho - \sigma|$ to denote the trace distance between $\rho$ and $\sigma$ and $F (\rho, \sigma) = \Tr (|\sqrt{\rho} \sqrt{\sigma}|)$ for the square root fidelity of two mixed states, $\rho$ and $\sigma$.

Letting $X,Y,Z$ denote the single-qubit Pauli matrices, the $n$-qubit Pauli group is
\[
    \mathcal P_n
    :=
    \left\{
        \omega\, P_1\otimes \cdots \otimes P_n
        \;:\;
        \omega\in\{\pm 1,\pm i\},\;
        P_j\in\{I,X,Y,Z\}
    \right\}.
\]
The $n$-qubit Clifford group is the normalizer of the Pauli group in the unitary group:
\[
    \mathcal C_n
    :=
    \left\{
        U\in U(2^n)
        \;:\;
        U\mathcal P_n U^\dagger = \mathcal P_n
    \right\}.
\]
Equivalently, $U\in\mathcal C_n$ if and only if $U P U^\dagger\in\mathcal P_n$ for every
$P\in\mathcal P_n$.

Whenever the Pauli or Clifford group is used as a finite group in an isomorphism problem, we identify operators that differ only by a scalar phase. Thus $\mathcal P_n$ and $\mathcal C_n$ denote their projective versions in that setting. This convention is natural for the promise, and the lift $U\mapsto U\otimes\overline U$ is independent of the chosen representative.

\subsection{Problems and complexity classes} 
We use the standard definitions for complexity classes like \BQP, \NP, \SZK, \textsf{QCSZK}, \QSZK, \textsf{QCMA} which can be found, for instance, in~\cite{aaronson2005complexity}. Note that the definition of \textsf{QCSZK} is given in~\cite {lockhart2017quantum} implicitly.

We will need the following definitions for state hidden subgroup problems and quantum state distinguishability.
\begin{definition}[$(\alpha,\beta)$-MixedStateHSP]
    Let $G$ be a finite group with a unitary representation $R: G \rightarrow U(d)$. Let $H < G$ be a hidden subgroup of $G$.
    The goal is to find $H$ 
    given a classical description of a quantum circuit $Q$ preparing a mixed state $\rho$ promised that:
    \begin{align*}
        & \forall U \in H, D(U\rho U^\dagger,\rho) \leq \alpha, \\
        & \forall U \notin H, D(U\rho U^\dagger, \rho) \geq \beta.
    \end{align*}
\end{definition}

\begin{definition}[$(\alpha,\beta)$-Quantum State Distinguishability]
    Given classical descriptions of quantum circuits $Q_0,Q_1$ preparing mixed states $\rho_0,\rho_1$, accept if 
    \[ \dfrac12\norm{\rho_0 - \rho_1}_1 \geq \beta \]
    and reject if
    \[ \dfrac12\norm{\rho_0 - \rho_1}_1 \leq \alpha, \]
    promised that one is the case.
\end{definition}

\begin{theorem}[\cite{watrous2002quantum}]
    For any $0 \leq \alpha < \beta \leq 1$ such that $\alpha < \beta^2$ the $(\alpha,\beta)$-Quantum State Distinguishability problem is complete for the complexity class $\QSZK$.
\end{theorem}

\subsection{Classical shadows}

We require the following fundamental result from \cite{huang2020predicting} on the prediction of many properties of a quantum state from a few measurements. For an unknown quantum state $\rho$, a single classical shadow sample is obtained by drawing a random unitary $U$ (from a unitary ensemble such as the Clifford group), measuring $U \rho U^\dagger$ in the computational basis to obtain an outcome $b$, and then constructing the estimator
\[
\hat{\rho} = \mathcal{M}^{-1}\!\left(U^\dagger |b\rangle\langle b| U\right),
\]
where $\mathcal{M}^{-1}$ is a known linear map depending on the unitary ensemble.

\begin{theorem}[Classical Shadows for Fidelity Estimation \cite{huang2020predicting}]
    \label{fact:shadows}
    Let $\varepsilon,\delta \in (0,1)$, and let
$\ket{\phi_1},\dots,\ket{\phi_M}$ be pure states on a
$d$-dimensional Hilbert space. Let
$\ket{\psi}$ be an unknown pure state. Then there is a
randomized measurement procedure using global Clifford
classical shadows such that
\[
N = O\!\left(\frac{\log(M/\delta)}{\varepsilon^2}\right)
\]
independent copies of $\rho$ suffice to  estimate each overlap $|\langle \psi \mid \phi_i\rangle|^2$ within additive error $\varepsilon$
with probability at least $1-\delta$ over the randomness of the chosen
measurements and their outcomes.
\end{theorem}

One might wonder if the theorem above could be generalized to estimating fidelities between $\rho$ and density matrices $\sigma_j$. The above theorem can be generalized to an efficient estimator for linear functions such as $\Tr (\rho \sigma_j)$. However, fidelity is a nonlinear function. As a matter of fact, estimating the fidelity between two mixed states specified by their generating circuit is, in general, hard for $\QSZK$~\cite{gilyen2022improved}.

\subsection{Background on continuous variable formalism}

We briefly recall the continuous-variable formalism relevant for the bosonic isomorphism problems considered in this work. An $n$-mode bosonic system is described by the Hilbert space
\begin{align}
\mathcal H_n \;=\; \ell^2(\mathbb N^n),
\end{align}
with Fock basis $\{ \ket{\mathbf m} : \mathbf m=(m_1,\dots,m_n)\in \mathbb N^n\}$, where
\begin{align}
\ket{\mathbf m}
=
\prod_{i=1}^n \frac{(a_i^\dagger)^{m_i}}{\sqrt{m_i!}} \ket{0}.
\end{align}
The creation and annihilation operators satisfy the canonical commutation relations
\begin{align}
[a_i,a_j^\dagger]=\delta_{ij},\qquad [a_i,a_j]=[a_i^\dagger,a_j^\dagger]=0,
\end{align}
and the total photon-number operator is $N=\sum_{i=1}^n a_i^\dagger a_i$. In particular, the number operator at site $i$ is defined by $N_i = a^\dag_i a_i$. We also have the position and momentum operators defined as $X = (a+a^\dag)/\sqrt 2$ and $P=(a-a^\dag)/(i\sqrt 2)$.

A convenient description of pure bosonic states is given by the Segal--Bargmann (or stellar) representation. For
\begin{align}
\ket{\psi}=\sum_{\mathbf m\in\mathbb N^n}\psi_{\mathbf m}\ket{\mathbf m},
\end{align}
its stellar function is the holomorphic function
\begin{align}
F_\psi(\mathbf z)
= \exp(|\mathbf z|^2/2) \bra{\mathbf z^\ast}\psi\rangle = 
\sum_{\mathbf m\in\mathbb N^n}
\psi_{\mathbf m}\,
\frac{\mathbf z^{\mathbf m}}{\sqrt{\mathbf m!}},
\qquad
\mathbf z=(z_1,\dots,z_n)\in\mathbb C^n,
\end{align}
where $\mathbf z^{\mathbf m}:=\prod_{i=1}^n z_i^{m_i}$ and $\mathbf m!:=\prod_{i=1}^n m_i!$. In this representation, creation and annihilation operators act as multiplication and differentiation,
\begin{align}
a_i^\dagger \;\leftrightarrow\; \times z_i,
\qquad
a_i \;\leftrightarrow\; \partial_{z_i}.
\end{align}

A pure state $\ket{\psi}$ is said to have \emph{finite stellar rank} $r$ if its stellar function can be written as
\begin{align}
F_\psi(\mathbf z)=P(\mathbf z)\,G(\mathbf z),
\end{align}
where $P$ is a polynomial of total degree $r$ and $G$ is a Gaussian holomorphic function. The case $r=0$ corresponds exactly to Gaussian states. Equivalently, any finite-stellar-rank state may be written as
\begin{align}
\ket{\psi}=U_G\ket{c},
\end{align}
where $U_G$ is a Gaussian unitary and $\ket{c}$ is a \emph{core state}, namely a state with finite support in the Fock basis. Thus, if $\ket{c}$ has support only on basis states with total photon number at most $r$, then its stellar function is simply a polynomial of degree at most $r$. This makes the core states the natural finite-dimensional representatives of Gaussian equivalence classes. We note that for a core state, the stellar function is merely a polynomial, i.e.,
\begin{align}
F_c(\mathbf z) = P(\mathbf z).
\end{align}

Gaussian unitaries are the unitaries generated by Hamiltonians that are at most quadratic in the operators $a_i,a_i^\dagger$. Equivalently, they are exactly the unitaries that map Gaussian states to Gaussian states. On mode operators, they act by Bogoliubov transformations, i.e.,
\begin{align}
V_G\, a_i\, V_G^\dagger
=
\sum_{j=1}^n
\left( X_{ij} a_j + Y_{ij} a_j^\dagger \right) + \alpha_i,
\end{align}
for suitable matrices $X,Y$ and displacement vector $\alpha$. The subclass of \emph{passive} or \emph{linear optical} Gaussian unitaries consists of the number-preserving quadratic evolutions, generated by Hamiltonians of the form
\begin{align}
H=\sum_{i,j=1}^n h_{ij} a_i^\dagger a_j,
\end{align}
for a Hermitian matrix $(h_{ij})_{ij}$.
These are precisely the linear optical unitaries. They act on the creation operators by a unitary mixing of modes,
\begin{align}
V\, a_i^\dagger\, V^\dagger
=
\sum_{j=1}^n U_{ji} a_j^\dagger,
\qquad U\in U(n),
\end{align}
and therefore induce a unitary linear change of variables on the associated stellar polynomial. In particular, if $\ket{c}=P_c(a^\dagger)\ket{0}$ is a core state, then $V\ket{c}$ is again a core state whose polynomial is obtained from $P_c$ by the substitution $\mathbf z\mapsto U\mathbf z$. For a group element $U\in\mathrm{U}(n)$, we use the notation $R(U) = V\in\mathcal O_n$ to denote its corresponding unitary on the Hilbert space of $n$ modes. This is exactly the notion of equivalence underlying the linear optical isomorphism problem studied in this work. We refer to \cite{chabaud2020stellar,chabaud2021continuous, chabaud2022holomorphic,serafini2023quantum} for a more elaborate discussion on these topics. 

\paragraph{Equivalent reformulations}
When studying the problem of the isomorphism of core states under linear optical transformations, we are essentially asking whether for two polynomials $P_1, P_2\in\mathbb C[z_1,\cdots,z_n]$ there exists a unitary $U\in\mathrm{U}(n)$ such that $P_1(\mathbf z)$ is close to $P_2(U\mathbf z)$, in a space where monomials are orthogonal and $\norm{\mathbf z^{\mathbf k}} = \sqrt{\mathbf k!} $.

We highlight that this problem can also be viewed as an isomorphism over qudits with large $d$ (i.e., polynomial). In particular, deciding whether two $n$-mode bosonic core states of stellar rank $r$ are isomorphic under a linear optical transformation is equivalent to asking whether two symmetric states over $(\mathbb {C}^{n})^{\otimes r}$ are isomorphic under the product unitary group. In other words, imagine we are given $\ket{\psi_1},\ket{\psi_2}\in \mathrm{Symm}\left((\mathbb C^n)^{\otimes r}\right)$. The question is, whether there exists $U\in\mathrm{U}(n)$ such that $\ket{\psi_1} \approx U^{\otimes r}\ket{\psi_2}$. This connection was observed in \cite{migdal2014multiphoton}. We can make this connection more rigorous with the following map: For any Fock state $\ket{\mathbf k}\in \ell^2(\mathbb N^n)$, with $|\mathbf k| = r$, define the $r$-qu$d$it state with $d=n$ as
\begin{align}
\ket{s_{\mathbf k}} := \sqrt{\frac{r!}{\mathbf k!}} \, P_{\mathrm{sym}}^{n,r}\left( \ket{1}^{\otimes k_1} \ket{2}^{\otimes k_2}\cdots\ket{n}^{\otimes k_n} \right) \in \vee^{r}\mathbb C^n,
\end{align}
where $P_{\mathrm{sym}}^{n,r}$ is the projection into the symmetric subspace of $(\mathbb C^n)^{\otimes r}$ (denoted by $\vee^r \mathbb C^n$ following the notation of \cite{harrow2013church}). Note that the prefactor $\sqrt{\frac{r!}{\mathbf k!}}$ is to ensure that $\ket{s_{\mathbf k}}$ is normalized.

Using this notation, we highlight an elementary fact connecting the two problems.

\begin{fact}\label{fact:bosons}
For Fock states $\ket{\mathbf k}$ and $\ket{\mathbf j}$ with $|\mathbf j| = |\mathbf k| = r$, it is the case that
\begin{align}
\bra{\mathbf k} R(U) \ket{\mathbf j} = \bra{s_{\mathbf k}} U^{\otimes r}\ket{s_{\mathbf j}},
\end{align}
for any $U\in\mathrm{U}(n)$ and its corresponding linear optical unitary $R(U)$.
\end{fact}

\section{Pure State Group Isomorphism Problems}
In this section, we address the complexity of \textsc{PSGI} for various groups, which had previously been studied only for subgroups of the symmetric group under the representation that permutes qubits \cite{lockhart2017quantum}. In \cref{subsec:qcszk} we prove that all \textsc{PSGI} Problems are in \textsf{QCSZK}. In \cref{subsec:cip-gi-hardness} we prove that the Pure State Clifford Isomorphism Problem is at least as hard as Graph Isomorphism. In \cref{subsec:pauli-bqp-completeness} we show that the Pure State Pauli Isomorphism Problem is \textsf{BQP}-complete.

\subsection{$\QCSZK$ Containment of Pure State Group Isomorphism}
\label{subsec:qcszk}

In this section, we prove that the Pure State Group Isomorphism Problem is in \textsf{QCSZK}. The main insight is that the quantum messages in the \textsf{QSZK} protocol given in \cite{lockhart2017quantum} can be replaced by classical shadows. We now restate the theorem that we will prove in this section.

\begin{theorem} [Restatement of the upper bound in \cref{thm:PSIG}]
    \label{thm:qcszk}
    Let $G$ be any efficiently represented finite group. Then $(\alpha,\beta)$-\textsc{PSGI}$[G]$ with constant $\alpha$ and $\beta = 1 -1/\omega(\log |G|)$ is in \textsf{QCSZK}.
\end{theorem}
\begin{proof}
    We will describe a \textsf{QCSZK} protocol for the \textit{non}-isomorphism problem. One round of the protocol is as follows, where (V) denotes the Verifier and (P) denotes the Prover. The prover and verifier will repeat the protocol to amplify soundness. The zero-knowledge property follows from standard gap amplification arguments.

    Let $\ket{\psi_1}$ and $\ket{\psi_2}$ be the input states, that is, $\ket{\psi_1} = C_1\ket{0^n}$ and $\ket{\psi_2} = C_2\ket{0^n}$.
The \textsf{QCSZK} protocol is as follows:

\begin{enumerate}
        \item (V) Pick $j \in \{1,2\}$ uniformly at random, and pick a uniformly random group element $g$. 
        \item (V) Prepare $N = \Theta(\log|G|)$ classical shadows for the state $\ket{\psi'} = R(g)\ket{\psi_j}$. Send the description of the classical shadows to (P).
        \item (P) Send $j' \in \{1,2\}$ to (V).
        \item (V) Accept if and only if $j' = j$.
    \end{enumerate}

Now, let's consider the probability that Prover (P) can determine the correct value of $j'$. To do this, we can examine the total variation distance between the distributions of classical shadows from step (2). Let $D_1$ and $D_2$ be the distributions of classical shadows for $j=1$ and $j=2$, respectively. Furthermore let $\rho_1$ and $\rho_2$ be the pre-measurement states:
\[\rho_1 = \mathcal{E}(\ket{\psi_1}) = \frac{1}{|G|}\sum_{g \in G}\left( R(g)\ketbra{\psi_1}{\psi_1}R(g)^\dagger \right)^{\otimes N} \]
and
\[\rho_2 = \mathcal{E}(\ket{\psi_2}) = \frac{1}{|G|}\sum_{g \in G}\left( R(g)\ketbra{\psi_2}{\psi_2}R(g)^\dagger \right)^{\otimes N} \]

Let's consider a NO instance of the non-isomorphism problem, that is, there exists $h \in G$ such that
\[
    \abs{\bra{\psi_1}R(h)\ket{\psi_2}}\geq \beta.
\]
First, note that $\mathcal{E}$ is the $G$-twirling channel, and therefore 
\[\rho_2 = \mathcal{E}(\ket{\psi_2}) = \mathcal{E}(R(h)\ket{\psi_2}).\]
Now, since fidelity is multiplicative under the tensor product, and by the monotonicity of fidelity, 
\[F(\rho_1,\rho_2) \geq \beta^N.\]
Thus, the trace distance, and therefore TV distance on measurement outcomes, is bounded by 
\[T(\rho_1,\rho_2) \leq \sqrt{1 - \beta^{2N}}.\]
Therefore, since $N = \Theta(\log|G|)$, and $\beta = 1 - \dfrac{1}{\omega(\log|G|)}$, 
\[T(\rho_1,\rho_2) \leq o(1),\]
and thus the probability that the Prover can distinguish $j=1$ from $j=2$ is negligible.

Now, let's consider a YES instance, that is, for all $a \in G$,
\[ |\bra{\psi_1}R(a)\ket{\psi_2}| 
\leq \alpha.\]
Prover receives $N$ copies of $\ket{\psi'} = R(g)\ket{\psi_j}$.

By the guarantee of the classical shadows procedure (\cref{fact:shadows}), when $N = \Omega(\log(\frac{|G|}{\gamma})/\eps^2)$ where $\eps = (1-\alpha^2)/3$, the Prover can estimate the quantity
\[ |\bra{\psi_1}R(a)\ket{\psi'}|^2 = |\bra{\psi_1}R(a)R(g)\ket{\psi_j}|^2\]
for all $a \in G$ with probability $\geq 1 -\gamma$.
If $j=1$, then there exists $a \in G$, namely $a = g^{-1}$, such that 
\[
|\bra{\psi_1}R(a)R(g)\ket{\psi_j}|^2 = 1,
\]
and if $j=2$, then for all $a$,
\[
|\bra{\psi_1}R(a)R(g)\ket{\psi_j}|^2 \leq \alpha^2.
\]
\end{proof}
Thus, Prover succeeds in learning $j$ with probability $\geq 1- \gamma$.

\subsection{Pure State Group Isomorphism is $\BQP$-hard}
\label{sec:PSGI-BQP-hard}
In this section, we show that PSGI$[G]$ is $\BQP$-hard for all groups. We will reduce from the $\BQP$-complete problem of deciding whether a quantum circuit $Q$ accepts or rejects on input $\ket{0^n}$. In particular, standard amplification gives a $\BQP$-hard promise problem in which, for any efficiently preparable $\ket\phi$ satisfying $\abs{\braket{\phi}{0^n}}<\eps$, one decides which of the following overlaps is large:
\[  
    \abs{\bra{0^n}Q\ket{0^n}} \geq 1 - \eps
\]
or
\[  
    \abs{\bra{\phi}Q\ket{0^n}} \geq 1 - \eps.
\]
Furthermore the problem remains $\BQP$-hard for $\eps = \exp(-n)$.

We now prove the theorem.
\begin{theorem}[Restatement of the $\BQP$-hardness in \cref{thm:PSIG}]
    \label{thm:bqp-hardness-psgi}
    $(\alpha,\beta)$-\textsc{PSGI}$[G]$ is $\BQP$-hard for all nontrivial finite groups $G$ under randomized reductions, even when $\alpha = o(1), \beta > 1 - \exp(-n)$.
\end{theorem}
\begin{proof}
    We reduce from the $\BQP$-complete problem of deciding whether a quantum circuit $Q$ accepts or rejects, as specified above.
    In particular, we take $Q'$ such that if $Q$ rejects then
    \[
        \abs{\bra{0^n}Q'\ket{0^n}} \geq 1 - \exp(-n)
    \]
    and if $Q$ accepts then
    \[
        \abs{\bra{\phi}Q'\ket{0^n}} \geq 1 - \exp(-n),
    \]
    where $\ket{\phi}$ satisfies
    \[
        \max_{g \in G} \abs{\bra{\phi}R(g)\ket{0^n}} \leq o(1).
    \]
    We define our PSGI$[G]$ instance as 
    \begin{align*}
        \ket{\psi_1} &= Q'\ket{0^n} \\
        \ket{\psi_2} &= \ket{\phi}.
    \end{align*}
    Clearly, if $Q$ accepts, then we have a YES instance of PSGI$[G]$, and if $Q$ rejects, then we have a NO instance. 
    Finally, it remains to show how to choose $\ket{\phi}$. Fortunately, we can take $\ket{\phi}$ to be the output of an approximate state $t$-design by \cref{claim:maxfidelity}.
\end{proof}

Note that the above reduction uses a randomized reduction to sample the circuit preparing the state $\ket{\phi}$, but for many fixed groups, one can choose $\ket{\phi}$ deterministically. 

\subsection{Pure State Group Isomorphism Reduces to StateHSP for Abelian Groups} \label{subsec:psgi-abelian-hsp}

In this section, we show that for abelian $G$, PSGI$[G]$ reduces to StateHSP over the generalized dihedral group $G \ltimes \mathbb Z_2$. Importantly, this means that there is an efficient quantum algorithm for PSGI$[G]$ for groups such that the hidden subgroup problem over $G \ltimes \mathbb Z_2$ admits an efficient quantum algorithm. Our reduction follows similar ideas as the reduction from the hidden shift problem to the hidden subgroup problem \cite{friedl2014hidden}, but generalized to work for the state versions.

Note that elements of $G \ltimes \mathbb Z_2$ are described by tuples $(g,a)$ where $g \in G, a \in \mathbb{Z}_2$, and the group operation is given by 
\begin{align*}
    (g,a)(h,b) &= (g + (-1)^{a}h, a+b).
\end{align*}
Formally, we show the following.

\begin{lem}[Phase-invariant lift]\label{lem:phase-invariant-lift}
    Let $R$ be an efficient unitary or projective unitary representation of $G$.
    Given circuits preparing $\ket{\psi_1}$ and $\ket{\psi_2}$, define
    \[
        \ket{\psi_i^\star}
        :=\ket{\psi_i}\otimes\ket{\overline{\psi_i}},
        \qquad
        R^\star(g):=R(g)\otimes\overline{R(g)}.
    \]
    Then $R^\star$ is an ordinary unitary representation and, for every $g\in G$,
    \[
        \bra{\psi_1^\star}R^\star(g)\ket{\psi_2^\star}
        =\abs{\bra{\psi_1}R(g)\ket{\psi_2}}^2.
    \]
    In particular, the lifted overlap is nonnegative and an $(\alpha,\beta)$
    instance becomes an $(\alpha^2,\beta^2)$ instance under
    the lifted representation.
\end{lem}

\begin{proof}
    If $R(g)R(h)=\omega(g,h)R(gh)$, where $\abs{\omega(g,h)}=1$, then
    the two phases in $R^\star(g)R^\star(h)$ multiply to
    $\omega(g,h)\overline{\omega(g,h)}=1$. Thus $R^\star$ is an ordinary
    representation. Moreover,
    \begin{align*}
        \bra{\psi_1^\star}R^\star(g)\ket{\psi_2^\star}
        &=\bra{\psi_1}R(g)\ket{\psi_2}
          \bra{\overline{\psi_1}}\overline{R(g)}
          \ket{\overline{\psi_2}}\\
        &=\bra{\psi_1}R(g)\ket{\psi_2}\,
          \overline{\bra{\psi_1}R(g)\ket{\psi_2}}\\
        &=\abs{\bra{\psi_1}R(g)\ket{\psi_2}}^2.
    \end{align*}
    The state $\ket{\overline{\psi_i}}$ is prepared by replacing every gate
    in the classical description of $C_i$ by its entrywise conjugate; the same
    gate-by-gate construction implements $\overline{R(g)}$. Hence the lift has
    polynomial overhead. This construction crucially uses circuit description, rather than
    black-box copies of unknown states.
\end{proof}

\begin{theorem} \label{thm:abelianreduction}
    Let $G$ be a finite abelian group with an efficient unitary or projective unitary representation
    $R:G\to U(2^n)$. Then $(\alpha, 1)$-PSGI$[G]$ reduces to
    StateHSP over the generalized dihedral group $G\rtimes \mathbb Z_2$.
\end{theorem}

\begin{proof}
    Given an instance of PSGI$[G]$ with states $\ket{\psi_1},\ket{\psi_2}$,
    form the lifted states and representation from
    \cref{lem:phase-invariant-lift}, and define
    \[
        \ket{\Xi}
        =
        \frac{1}{\sqrt 2}\left(\ket{0}\ket{\psi_1^\star}
        +
        \ket{1}\ket{\psi_2^\star}\right).
    \]
    Let $\Gamma=G\rtimes \mathbb Z_2$, with group law
    \[
        (g,a)(h,b)=(g+(-1)^a h,a+b),
    \]
    where addition in the second register is modulo $2$. Define a representation
    $R_\star':\Gamma\to U(2^{2n+1})$ by
    \[
        R_\star'(g,0)
        =
        \ketbra{0}\otimes R^\star(g)
        +
        \ketbra{1}\otimes R^\star(-g),
        \qquad
        R_\star'(0,1)=X\otimes I,
    \]
    and
    \[
        R_\star'(g,a)=R_\star'(g,0)R_\star'(0,a).
    \]
    This is a valid representation because
    \[
        R_\star'(g,0)R_\star'(h,0)=R_\star'(g+h,0),
    \]
    and
    \[
        R_\star'(0,a)R_\star'(h,0)
        =R_\star'((-1)^a h,0)R_\star'(0,a).
    \]
    Hence
    \[
        R_\star'(g,a)R_\star'(h,b)
        =
        R_\star'(g+(-1)^a h,a+b).
    \]
Let us also comment on the efficiency of computing $R_\star'.$ The conjugate
circuits have polynomial size by \cref{lem:phase-invariant-lift}. Aside from
those circuits, the only additional group operation is computing $-g$ from
$g$. Whenever an explicit presentation of $G$ includes the identity element,
we can compute $-g$ efficiently by subtracting $g$ from the identity. In general, it may
be necessary to compute the order of $g$; because $G$ is abelian, this can be
done in quantum polynomial time. Thus the reduction is classical whenever
group inverses are classically efficient, and otherwise it is quantum.

    Next, let
    \[
        H
        =
        \{x\in \Gamma:R_\star'(x)\ket{\Xi}=\ket{\Xi}\}
    \]
    be the set of group elements stabilizing $\ket{\Xi}$.
    We claim that $H$ contains an element of the form $(h,1)$ if and only if the
    original PSGI instance is a YES instance.

    For any $h\in G$, we have
    \[
    \begin{aligned}
        \bra{\Xi}R_\star'(h,1)\ket{\Xi}
        &=
        \frac{1}{2}
        \left(
            \bra{\psi_1^\star}R^\star(h)\ket{\psi_2^\star}
            +
            \bra{\psi_2^\star}R^\star(-h)\ket{\psi_1^\star}
        \right)  \\
        &=
        \abs{\bra{\psi_1}R(h)\ket{\psi_2}}^2.
    \end{aligned}
    \]
    Therefore, if $(h,1)\in H$, then
    \[
        \abs{\bra{\psi_1}R(h)\ket{\psi_2}}=1,
    \]
    so the PSGI instance is a YES instance.

    Conversely, if the PSGI instance has perfect completeness, then there exists
    $h\in G$ and a phase $e^{i\theta}$ such that
    \[
        R(h)\ket{\psi_2}=e^{i\theta}\ket{\psi_1}.
    \]
    The lift cancels this phase, and hence
    \[
        R^\star(h)\ket{\psi_2^\star}=\ket{\psi_1^\star},
        \qquad
        R^\star(-h)\ket{\psi_1^\star}=\ket{\psi_2^\star}.
    \]
    Thus
    \[
    \begin{aligned}
        R_\star'(h,1)\ket{\Xi}
        &=
        R_\star'(h,0)(X\otimes I)
        \frac{1}{\sqrt 2}
        \left(\ket{0}\ket{\psi_1^\star}
        +\ket{1}\ket{\psi_2^\star}\right) \\
        &=
        \frac{1}{\sqrt 2}
        \left(
            \ket{0}R^\star(h)\ket{\psi_2^\star}
            +
            \ket{1}R^\star(-h)\ket{\psi_1^\star}
        \right) \\
        &=
        \frac{1}{\sqrt 2}
        \left(\ket{0}\ket{\psi_1^\star}
        +\ket{1}\ket{\psi_2^\star}\right) \\
        &=
        \ket{\Xi}.
    \end{aligned}
    \]
    Hence $(h,1)\in H$.

    Therefore, after solving StateHSP over $\Gamma$, it suffices to check whether
    the returned subgroup has an element with second component $1$. If the subgroup
    is given by generators, this is equivalent to checking whether at least one
    generator has second component $1$, since the projection
    $\Gamma\to \mathbb Z_2$ is a group homomorphism.
\end{proof}

Importantly, there are several common groups which admit efficient algorithms for StateHSP over $G \ltimes \mathbb{Z}_2$, such as $\mathbb{Z}_2^n$ and the abelianization of the Pauli group.
For a more thorough treatment of groups that admit efficient HSP algorithms, see \cite{friedl2014hidden}, and for algorithms for the StateHSP problem, see \cite{hinsche2025abelian}. Our reduction also implies that one can solve PSGI$[G]$ in subexponential time for groups whose corresponding StateHSP instance can be solved by  Kuperberg's algorithm \cite{kuperberg2005subexponential}. In particular, this yields a subexponential time quantum algorithm for PSGI$[\mathbb{Z}_N]$ as it reduces to StateHSP over the dihedral group, $\mathbb{Z}_N \ltimes \mathbb{Z}_2.$

We can also ``robustify'' this statement to allow for imperfect PSGI completeness and soundness. Formally, we show the following.

\begin{theorem}[Robust reduction from PSGI to approximate StateHSP] \label{thm:robustabelian}
    Let $G$ be a finite abelian group with an efficient unitary or projective unitary representation
    $R:G\to U(2^n)$, and let $\alpha<1$ be a constant. Let $m=m(n)$ be
    polynomially bounded.

    Suppose an instance of PSGI$[G]$ is promised to satisfy one of the following:
    \begin{itemize}
        \item[(YES)] There exists $h\in G$ such that
        \[
            \abs{\bra{\psi_1}R(h)\ket{\psi_2}}
            \ge \beta = 1-\varepsilon.
        \]
        \item[(NO)] For every $h\in G$,
        \[
            \left|\bra{\psi_1}R(h)\ket{\psi_2}\right|\le \alpha.
        \]
    \end{itemize}

    Then there is an efficient reduction to an approximate StateHSP instance over
    the generalized dihedral group $G\rtimes \mathbb Z_2$ such that:
    \begin{itemize}
        \item in the YES case, there exists an element $(h,1)\in G\rtimes \mathbb Z_2$
        that stabilizes the resulting state up to overlap at least
        \[
            \beta^{2m} = (1-\varepsilon)^{2m} \ge 1-2m\varepsilon;
        \]
        \item in the NO case, every element $(h,1)\in G\rtimes \mathbb Z_2$ has
        overlap at most
        \[
            \alpha^{2m}.
        \]
    \end{itemize}

    More generally, if $\varepsilon\le 1/q(n)$ and $2m(n)/q(n)\le 1/p(n)$ for some
    polynomial $p$, then the resulting instance has completeness at least
    $1-1/p(n)$ and soundness $\alpha^{2m}$. Choosing $m=\omega(\log n)$ gives
    negligible soundness.
\end{theorem}

\begin{proof}
    Using the notation of \cref{thm:abelianreduction}, consider $m$ copies of
    the lifted superposition $\ket{\Xi}$ and the $m$-fold product of its
    representation:
    \[
        \ket{\Phi}:=\ket{\Xi}^{\otimes m},
        \qquad
        R'_{\star,m}(g,a):=R_\star'(g,a)^{\otimes m}.
    \]
    This is still an efficient representation of $\Gamma$, since $m$ is polynomially
    bounded.

    For every $h\in G$, we have
    \[
    \begin{aligned}
        \bra{\Xi}R_\star'(h,1)\ket{\Xi}
        =\abs{\bra{\psi_1}R(h)\ket{\psi_2}}^2.
    \end{aligned}
    \]
    Therefore,
    \[
    \begin{aligned}
        \bra{\Phi}R'_{\star,m}(h,1)\ket{\Phi}
        &=
        \left(
            \bra{\Xi}R_\star'(h,1)\ket{\Xi}
        \right)^m \\
        &=\abs{\bra{\psi_1}R(h)\ket{\psi_2}}^{2m}.
    \end{aligned}
    \]

    In the YES case, there exists $h\in G$ such that
    \[
        \abs{\bra{\psi_1}R(h)\ket{\psi_2}}
        \ge 1-\varepsilon.
    \]
    Hence
    \[
        \bra{\Phi}R'_{\star,m}(h,1)\ket{\Phi}
        \ge
        (1-\varepsilon)^{2m}.
    \]
    By Bernoulli's inequality,
    \[
        (1-\varepsilon)^{2m} \ge 1-2m\varepsilon.
    \]
    Thus the odd element $(h,1)$ approximately stabilizes $\ket{\Phi}$ with
    overlap at least $1-2m\varepsilon$.

    In the NO case, for every $h\in G$,
    \[
        \left|\bra{\psi_1}R(h)\ket{\psi_2}\right|\le \alpha.
    \]
    Therefore, for every odd element $(h,1)$,
    \[
    \begin{aligned}
        \left|\bra{\Phi}R'_{\star,m}(h,1)\ket{\Phi}\right|
        &=\abs{\bra{\psi_1}R(h)\ket{\psi_2}}^{2m}\\
        &\le \alpha^{2m}.
    \end{aligned}
    \]
    Since $\alpha<1$ is constant, $\alpha^{2m}=\operatorname{negl}(n)$ whenever
    $m=\omega(\log n)$.

    Hence the reduction maps PSGI$[G]$ with completeness error $\varepsilon$ and
    constant soundness $\alpha$ to an approximate StateHSP instance over
    $G\rtimes\mathbb Z_2$ with odd-element completeness at least $(1-\varepsilon)^{2m}$
    and odd-element soundness at most $\alpha^{2m}$.
\end{proof}

We end this section by noting that the approximate StateHSP construction can be
combined with efficient Fourier-sampling algorithms for the corresponding
ordinary StateHSP. In elementary abelian cases, the total probability of a
faulty equation over the full sample budget must be $o(1)$. We make this
condition explicit and recover $(h,1)$ for the Pauli group in the next section.

\subsection{Pure State Pauli Isomorphism Problem is $\BQP$-complete}
\label{subsec:pauli-bqp-completeness}
The results of the previous two sections allow us to show that PSGI for the $n$-qubit Pauli group is $\mathsf{BQP}$-complete. Formally, we have that

\begin{theorem}[Restatement of \cref{thm:main-pauli-bqp-completenes}]
    Let $\alpha<1$ be a constant and let $\beta=1-o(1/(n\log^2 n))$. Then
    $(\alpha,\beta)$-PSGI$[\mathcal P_n]$ is $\mathsf{BQP}$-complete.
\end{theorem}

\begin{proof}
    $\BQP$-hardness follows immediately from \cref{thm:bqp-hardness-psgi}. In fact, for the Pauli group the reduction from any problem in $\BQP$ can be made deterministic by picking the $\ket{\phi}$ state in \cref{thm:bqp-hardness-psgi} to be, for instance, $\ket{\phi} = \ket{T^n}.$ This is sufficient, since for every Pauli $P \in \mathcal{P}_n$ it is the case that
    \[
    |\bra{\phi} P \ket{0^n}| = \operatorname{negl}(n).
    \]

    To show containment in $\BQP$, we combine the phase-invariant lift with the
    approximate StateHSP construction. Define
    \[
        \mathcal P_n^\star
        :=\{P\otimes\overline{P}:P\in\mathcal P_n\}.
    \]
    Starting with the full Pauli group defined in the preliminaries, the map
    $P\mapsto P\otimes\overline P$ has scalar kernel
    $\{\pm I,\pm iI\}$. Each distinct lifted Pauli is a Hermitian involution,
    and any commutation phase is canceled by the conjugate tensor factor.
    Its image therefore has $2^{2n}$ elements and
    \[
        \mathcal P_n^{\star}\cong\mathbb Z_2^{2n}.
    \]

    Let $\varepsilon:=1-\beta$, and define
    \[
        \ket{\phi_1}:=\ket{\psi_1}\ket{\overline{\psi_1}},
        \qquad
        \ket{\phi_2}:=\ket{\psi_2}\ket{\overline{\psi_2}}.
    \]
    Suppose we are in the YES case, so that there exists
    $P\in \mathcal P_n$ such that
    \[
        \abs{\bra{\psi_1}P\ket{\psi_2}}\ge \beta.
    \]
    Let $a:=\bra{\psi_1}P\ket{\psi_2}$. Then, for
    $Q:=P\otimes\overline P\in\mathcal P_n^\star$,
    \[
        \bra{\phi_1}Q\ket{\phi_2}
        =\abs{a}^2
        \ge\beta^2
        \ge 1-2\varepsilon.
    \]

    Conversely, in the NO case, for every $P\in\mathcal P_n$,
    \[
        \left|\bra{\psi_1}P\ket{\psi_2}\right|\le \alpha.
    \]
    Therefore, for every $Q=P\otimes\overline P\in\mathcal P_n^\star$,
    \[
        \left|\bra{\phi_1}Q\ket{\phi_2}\right|
        =
        \left|\bra{\psi_1}P\ket{\psi_2}\right|^2
        \le
        \alpha^2.
    \]
    Hence the lifted instance has soundness at most $\alpha^2<1$.

    Apply the superposition-and-amplification construction from
    \cref{thm:robustabelian} to this already phase-invariant abelian
    representation. Let
    \[
        m:=\lceil \log^2 n\rceil.
    \]
    The reduction produces the state
    \[
        \ket{\Psi}
        =
        \frac{1}{\sqrt 2}
        \left(
            \ket{0}\ket{\phi_1}
            +
            \ket{1}\ket{\phi_2}
        \right),
    \]
    and then takes
    \[
        \ket{\Phi}:=\ket{\Psi}^{\otimes m}.
    \]
    The relevant group is
    \[
        \Gamma:=\mathcal P_n^\star\times \mathbb Z_2.
    \]
    The inversion action is trivial because every lifted Pauli has order two,
    so this direct product is the corresponding generalized-dihedral group. In
    particular,
    \[
        \Gamma\cong \mathbb Z_2^{2n+1}.
    \]
    Define $\rho(Q,b):=X^b\otimes Q$ and let
    \[
        \rho_m(x):=\rho(x)^{\otimes m}
    \]
    be the corresponding representation on $\ket{\Phi}$.

    For $Q=P\otimes\overline P$,
    \[
        \bra{\Phi}\rho_m(Q,1)\ket{\Phi}
        =\abs{\bra{\psi_1}P\ket{\psi_2}}^{2m}.
    \]
    Therefore, in the YES case there exists an odd element $(Q,1)\in \Gamma$
    such that
    \[
        \bra{\Phi}\rho_m(Q,1)\ket{\Phi}
        \ge
        \beta^{2m}
        \ge
        1-2m\varepsilon,
    \]
    where the last inequality follows from Bernoulli's inequality. In the NO
    case, every odd element $(Q,1)\in\Gamma$ satisfies
    \[
        \left|\bra{\Phi}\rho_m(Q,1)\ket{\Phi}\right|
        \le
        \alpha^{2m}.
    \]
    Since $\alpha<1$ is constant and $m=\lceil\log^2 n\rceil$, we have
    \[
        \alpha^{2m}=\operatorname{negl}(n).
    \]

    It remains to explain how to solve this approximate StateHSP instance in
    quantum polynomial time. Since $\Gamma\cong \mathbb Z_2^{2n+1}$, its quantum
    Fourier transform is just the Hadamard transform over $O(n)$ qubits. We use
    the standard Fourier sampling procedure for StateHSP: prepare
    \[
        \frac{1}{\sqrt{|\Gamma|}}\sum_{x\in\Gamma}\ket{x}\ket{\Phi},
    \]
    apply the controlled unitary $\rho_m(x)$, apply the quantum Fourier transform
    over $\Gamma$ to the first register, and measure the first register. This
    produces a character $\chi\in \widehat{\Gamma}$.

    For $x\in \Gamma$, define
    \[
        f(x):=\bra{\Phi}\rho_m(x)\ket{\Phi}.
    \]
    Since $\Gamma$ is an elementary abelian $2$-group, each $\rho_m(x)$ is a
    Hermitian involution, and hence $f(x)\in[-1,1]$. Standard Fourier sampling
    gives
    \[
        \Pr[\chi(x)=1]
        =
        \frac{1+f(x)}{2}.
    \]
    Repeat the Fourier sampling procedure
    \[
        T:=C\log|\Gamma|=O(n)
    \]
    times, for a sufficiently large constant $C$, obtaining characters
    $\chi_1,\ldots,\chi_T$. Let
    \[
        L:=\{x\in \Gamma:\chi_i(x)=1\text{ for all }i=1,\ldots,T\}.
    \]
    Since $\Gamma\cong \mathbb Z_2^{2n+1}$, the set $L$ is a linear subspace and
    can be computed by Gaussian elimination over $\mathbb F_2$.

    In the YES case, let $k=(Q,1)$ be an odd element satisfying
    \[
        f(k)\ge 1-2m\varepsilon.
    \]
    Then one Fourier sample satisfies $\chi(k)=1$ with probability at least
    $1-m\varepsilon$. Hence
    \[
        \Pr[k\in L]
        \ge
        1-Tm\varepsilon.
    \]
    Thus, provided
    \[
        Tm\varepsilon=o(1),
    \]
    the true odd approximate stabilizer $k$ survives all the Fourier samples
    with high probability. Since $T=O(n)$ and $m=\lceil\log^2 n\rceil$, it
    suffices to assume, for example, that
    \[
        \varepsilon=o\left(\frac{1}{n\log^2 n}\right).
    \]

    Moreover, any odd element $x$ with $f(x)\le 2/3$ survives one Fourier sample
    with probability at most $5/6$. By a union bound over all odd elements of
    $\Gamma$, and by choosing the constant $C$ in $T=C\log|\Gamma|$ sufficiently
    large, with high probability no odd element $x$ satisfying $f(x)\le 2/3$
    remains in $L$. Therefore, in the YES case, with high probability $L$
    contains an odd element, and every odd element in $L$ has overlap greater
    than $2/3$.

    In the NO case, every odd element $x\in \Gamma$ satisfies
    \[
        |f(x)|\le \alpha^{2m}=\operatorname{negl}(n).
    \]
    In particular, for all sufficiently large $n$,
    \[
        f(x)\le 1/10.
    \]
    Hence any fixed odd element survives all $T$ Fourier samples with probability
    at most
    \[
        \left(\frac{1+1/10}{2}\right)^T.
    \]
    Again by a union bound over all odd elements, and by choosing $C$ sufficiently
    large, with high probability no odd element remains in $L$.

    The algorithm is therefore as follows. After computing $L$, check whether
    $L$ contains an element whose $\mathbb Z_2$ component is $1$. This can be done
    by Gaussian elimination. If no such element exists, reject. If such an odd
    element $x$ exists, estimate
    \[
        f(x)=\bra{\Phi}\rho_m(x)\ket{\Phi}
    \]
    to constant additive accuracy using the Hadamard test, and accept if the
    estimate is at least $1/2$.

    In the YES case, with high probability the algorithm finds an odd element
    $x\in L$ with $f(x)>2/3$, and the Hadamard test accepts. In the NO case, every
    odd element has $|f(x)|\le \alpha^{2m}=\operatorname{negl}(n)$, so the
    Hadamard test rejects except with small probability. Thus the containment
    algorithm has bounded error and runs in quantum polynomial time.

    This proves that $(\alpha,\beta)$-PSGI$[\mathcal P_n]$ is in $\BQP$ for
    constant $\alpha<1$ and $\beta=1-\varepsilon$ satisfying
    $\varepsilon=o(1/(n\log^2 n))$. In particular, this holds whenever
    $\beta\ge 1-1/n^c$ for any constant $c>1$. Together with the $\BQP$-hardness
    shown above, this proves $\BQP$-completeness in this parameter regime.
    
\end{proof}

\subsection{Pure State Clifford Isomorphism Problem is GI-hard}
\label{subsec:cip-gi-hardness}
In this section, we prove that the Pure State Clifford Isomorphism Problem is at least as hard as Graph Isomorphism (GI). 
\begin{theorem} [Restatement of GI hardness in \cref{thm:PSIG}]
    Graph Isomorphism on $n$-vertex graphs reduces to $(0.99999,1)$-PSGI$[\mathcal C_{n+1}]$.
\end{theorem}

Our goal is to prove that Graph Isomorphism reduces to the Pure State Clifford Isomorphism Problem. Let $G_1, G_2$ be two graphs. The main idea of our reduction is to construct the following states:
\begin{align*}
    \ket{\psi_1} = \frac{1}{\sqrt2} \ket{R}\ket{R^n} + \frac{1}{\sqrt2}\ket{R_-}\ket{G_1} \\
    \ket{\psi_2} = \frac{1}{\sqrt2} \ket{R}\ket{R^n} + \frac{1}{\sqrt2}\ket{R_-}\ket{G_2}
\end{align*}
where $\ket{G_1},\ket{G_2}$ are the graph states for $G_1, G_2$, respectively, $\ket{R} = \dfrac{1}{\sqrt2}\left( \ket{0} + e^{\pi i / 8}\ket{1} \right)$, $\ket{R_-} = \dfrac{1}{\sqrt2}\left( \ket{0} - e^{\pi i / 8}\ket{1} \right)$, and $\ket{R^n} = \ket{R}^{\otimes n}$. If $G_1$ and $G_2$ are isomorphic, then the states are exactly related by a permutation on qubits, which is an element of the Clifford group. It remains to show that if $\abs{\bra{\psi_1}C\ket{\psi_2}} \geq \beta$ for some Clifford $C$, then $C$ must be a permutation on qubits up to global phase. We proceed by the following lemma. For convenience we denote the first qubit as the \textit{control} and the remaining qubits as the \textit{target}.

\begin{lem}
    \label{lem:perm}
    Let $C \in \mathcal{C}_n$ be a Clifford unitary on $n$ qubits. If
    \[ \abs{\bra{R^n}C\ket{R^n}}^2 \geq 0.9999 \]
    where $\ket{R} = \frac{1}{\sqrt2}(\ket{0} + e^{i\pi /8}\ket{1})$,
    then $C$ is a permutation on qubits, up to global phase.
\end{lem}
\begin{proof}
    Let $P \in \mathcal{P}_n$ be a Pauli string. By H\"older's inequality and the Fuchs-van de Graaf inequality,
    \begin{align}
        \begin{split}
            \abs{\bra{R^n}P\ket{R^n} - \bra{R^n}C^\dagger P C\ket{R^n}} &=
            \abs{\Tr[P(\ketbra{R^n}{R^n} - C\ketbra{R^n}{R^n}C^\dagger)]} \\
            &\leq \norm{\ketbra{R^n}{R^n} - C\ketbra{R^n}{R^n}C^\dagger}_1 \norm{P}_\infty \\
            &\leq 2\sqrt{1 - \abs{\bra{R^n}C\ket{R^n}}^2} \\
            &\leq 2\sqrt{1 - 0.9999} = 0.02
        \end{split}
    \end{align}

    Let $X_i$ be the $X$ operator on the $i$-th qubit. Evaluating the expectation value yields $\bra{R^n}X_i\ket{R^n} = \cos(\pi/8) \approx 0.9238$. By the reverse triangle inequality and the fact that $\bra{R^n}C^\dagger P C\ket{R^n}$ is real-valued,
    \begin{align}
        \label{eq:R-state-fidelity}
        \cos(\pi/8) - 0.02 \leq \bra{R^n}C^\dagger X_i C\ket{R^n} \leq \cos(\pi/8) + 0.02.
    \end{align}

    Let $P' = C^\dagger X_iC$. Since $C$ is Clifford, $P' \in \mathcal{P}_n$. The inequality requires $\bra{R^n}P'\ket{R^n} \geq 0.9038$. It is easy to check that expectation value of any Pauli string $P'$ on $\ket{R^n}$ is
    \[ \bra{R^n}P'\ket{R^n} = s \cdot (0)^{|P'|_Z} \cdot \left(\cos(\pi/8)\right)^{|P'|_X} \cdot \left(\sin(\pi/8)\right)^{|P'|_Y} \]
    where $s \in \{1, -1\}$. Since $0.9038 > 0$, we require $s = 1$ and $|P'|_Z = 0$. Evaluating the possible configurations for $|P'|_X + |P'|_Y \geq 1$:
    \begin{itemize}
        \item If $|P'|_X = 1$ and $|P'|_Y = 0$, the expectation value is $\cos(\pi/8) \approx 0.9238$.
        \item If $|P'|_X = 2$ and $|P'|_Y = 0$, the expectation value is $\cos^2(\pi/8) \approx 0.8535$.
    \end{itemize}

    Since $0.8535 < 0.9038$, any configuration with weight greater than 1 or containing $Y$ operators strictly decreases the expectation value below the lower bound. Thus, $|P'|_X = 1$, $|P'|_Y = 0$, and $|P'|_Z = 0$. This implies $P' = X_{\pi(i)}$ for some index $\pi(i)$.

    Next, consider $Y_i$, where $\bra{R^n}Y_i\ket{R^n} = \sin(\pi/8) \approx 0.3826$. Let $Q' = C^\dagger Y_iC$. Applying the trace distance bound yields
    \[ \sin(\pi/8) - 0.02 \leq \bra{R^n}Q'\ket{R^n} \leq \sin(\pi/8) + 0.02 \]

    This requires $\bra{R^n}Q'\ket{R^n} \geq 0.3626$. Since $0.3626 > 0$, $Q'$ has a positive phase ($s=1$) and $|Q'|_Z = 0$. Because conjugating by Cliffords preserves commutation relations, $C^\dagger X_i C$ and $C^\dagger Y_i C$ must anticommute:
    \[ X_i Y_i = -Y_i X_i \implies X_{\pi(i)} Q' = -Q' X_{\pi(i)} \]

    For $Q'$ to anticommute with $X_{\pi(i)}$ while containing no $Z$ operators, $Q'$ must contain a $Y$ operator at index $\pi(i)$. Therefore, $Q' = Y_{\pi(i)} \otimes X_A$, where $X_A$ is a tensor product of $X$ operators on a subset of indices $A$ not containing $\pi(i)$. The expectation value of $Q'$ is $\sin(\pi/8)(\cos(\pi/8))^{|A|}$.

    If $|A| \geq 1$, the maximum expectation value is $\sin(\pi/8)\cos(\pi/8) \approx 0.3535$. Since $0.3535 < 0.3626$, we must have $|A| = 0$. Thus, $C^\dagger Y_i C = Y_{\pi(i)}$.

    Finally, consider $Z_i$:
    \[ C^\dagger Z_iC = C^\dagger (-iX_iY_i)C = -i (C^\dagger X_iC)(C^\dagger Y_i C) = -iX_{\pi(i)}Y_{\pi(i)} = Z_{\pi(i)} \]

    Since conjugation by $C$ permutes all single-qubit Paulis according to $\pi$,
    $C=e^{i\theta}U_\pi$ for some global phase $e^{i\theta}$ and the qubit
    permutation $U_\pi$.
\end{proof}

Now, we will continue with the reduction.
\begin{claim}
    \label{claim:reduction}
    Let $C \in \mathcal{C}_{n+1}$ be an $(n+1)$-qubit Clifford unitary and let
    \begin{align*}
        &\ket{\psi_1} = \frac{1}{\sqrt2}\ket{R,R^n} + \frac{1}{\sqrt2}\ket{R_-,G_1} \\
        & \ket{\psi_2} = \frac{1}{\sqrt2}\ket{R,R^n} + \frac{1}{\sqrt2}\ket{R_-,G_2}.
    \end{align*}
    If \[\abs{\bra{\psi_1}C\ket{\psi_2}} \geq 0.99999,\]
    then for $n \geq n_0$, for a sufficiently large constant $n_0,$ we have that
    \[\abs{\bra{R,R^n}C\ket{R,R^n}}^2 \geq 0.9999.\]
\end{claim}
\begin{proof}
    \begin{align*}
        0.99999 &\leq \abs{\bra{\psi_1}C\ket{\psi_2}}  \\ 
        &\leq \frac{1}{2}\abs{(\bra{R,R^n}C\ket{R,R^n} + \bra{R_-,G_1}C\ket{R_-,G_2} + \bra{R,R^n}C\ket{R_-,G_2} + \bra{R_-,G_1}C\ket{R,R^n})} \\
        &\leq \frac{1}{2}(\abs{\bra{R,R^n}C\ket{R,R^n}} + 1 + 2s)
    \end{align*}
    where $s = \exp(-n)$ is the maximum fidelity of $\ket{R,R^n}$ with any stabilizer rank $2$ state, since $\ket{R_-,G_1}, \ket{R_-,G_2}$ have stabilizer rank $2$.
    Thus, 
    \[\abs{\bra{R,R^n}C\ket{R,R^n}} \geq 2 \times 0.99999 -1 - 2s.\]
    Since $s = o(1)$,
    \[\abs{\bra{R,R^n}C\ket{R,R^n}}^2 \geq 0.9999.\]
\end{proof}

Now, combining \cref{claim:reduction} and \cref{lem:perm}, we see that if there exists a Clifford unitary $C$ such that
\[\abs{\bra{\psi_1}C\ket{\psi_2}} \geq 0.99999,\]
then $C$ must be a permutation up to global phase.

We need one more claim, that $C$ must leave the control qubit invariant.

\begin{claim}
    \label{claim:perm-first-qubit}
    Let $C$ be a permutation on $n+1$ qubits up to global phase and let
    \begin{align*}
        &\ket{\psi_1} = \frac{1}{\sqrt2}\ket{R,R^n} + \frac{1}{\sqrt2}\ket{R_-,G_1} \\
        & \ket{\psi_2} = \frac{1}{\sqrt2}\ket{R,R^n} + \frac{1}{\sqrt2}\ket{R_-,G_2}.
    \end{align*}
    If \[\abs{\bra{\psi_1}C\ket{\psi_2}} \geq 0.99999,\]
    then $C$ must leave the first qubit invariant.
\end{claim}
\begin{proof}
    We proceed similarly to the proof of \cref{claim:reduction},
    \begin{align*}
        0.99999 &\leq \abs{\bra{\psi_1}C\ket{\psi_2}}  \\ 
        &\leq \frac{1}{2}\abs{(\bra{R,R^n}C\ket{R,R^n} + \bra{R_-,G_1}C\ket{R_-,G_2} + \bra{R,R^n}C\ket{R_-,G_2} + \bra{R_-,G_1}C\ket{R,R^n})} \\
        &\leq \frac{1}{2}(\abs{\bra{R_-,G_1}C\ket{R_-,G_2}} + 1 + 2s),
    \end{align*}
    where $s = \exp(-n)$. 
    Thus, 
    \[
        \abs{\bra{R_-,G_1}C\ket{R_-,G_2}} \geq 0.9999
    \]
    Now suppose for the sake of contradiction that $C$ maps the first qubit to position $j \neq 1$. 
    Let $\Pi_j$ be the projector onto $\ket{R_-}$ on qubit $j$. Then
    \begin{align*}
        \abs{\bra{R_-,G_1}C\ket{R_-,G_2}}
        &=\abs{\bra{R_-,G_1}\Pi_j C\ket{R_-,G_2}}\\
        &\leq\sqrt{\bra{R_-,G_1}\Pi_j\ket{R_-,G_1}}\\
        &=\sqrt{\bra{R_-}\Tr_{\setminus j}(\ketbra{G_1}{G_1})\ket{R_-}}.
    \end{align*}
    Since $\ket{G_1}$ is a graph state, the one-qubit reduced density matrix $\Tr_{\setminus j}(\ketbra{G_1}{G_1})$
    is either a pure stabilizer state or $I/2$. In the first case its overlap
    with $\ket{R_-}$ is at most $\cos^2(\pi/16)$, and in the second it is
    $1/2$. Thus 
    $$\abs{\bra{R_-,G_1}C\ket{R_-,G_2}} \leq \cos(\pi/16)<0.9999,$$ which is a contradiction, and so $C$
    must leave the first qubit invariant.
\end{proof}

Finally, we can put it all together. Recall we already have one direction, i.e. if $G_1 \cong G_2$, then there exists a Clifford unitary $C$ such that
\[\abs{\bra{\psi_1}C\ket{\psi_2}} = 1 \geq 0.99999.\] 
So it remains to prove the other direction.

\begin{claim}
    Let
    \begin{align*}
        &\ket{\psi_1} = \frac{1}{\sqrt2}\ket{R,R^n} + \frac{1}{\sqrt2}\ket{R_-,G_1} \\
        & \ket{\psi_2} = \frac{1}{\sqrt2}\ket{R,R^n} + \frac{1}{\sqrt2}\ket{R_-,G_2},
    \end{align*}
    where $\ket{G_1},\ket{G_2}$ are graph states for $G_1,G_2$, respectively.
    If there exists a Clifford unitary $C$ such that
    \[\abs{\bra{\psi_1}C\ket{\psi_2}} \geq 0.99999,\]
    then $G_1 \cong G_2$.
\end{claim}
\begin{proof}
    Suppose there exists a Clifford unitary $C \in \mathcal{C}_{n+1}$ such that
    \[\abs{\bra{\psi_1}C\ket{\psi_2}} \geq 0.99999.\]
    
    By \cref{lem:perm} and \cref{claim:perm-first-qubit},
    $C=e^{i\theta}U_\pi$, where $U_\pi$ leaves the control qubit invariant.
    The scalar phase disappears from the magnitude, so we may use $U_\pi$ in
    place of $C$. Therefore
    \[U_\pi\ket{\psi_2} = \frac{1}{\sqrt2}\ket{R}\ket{R^n}+\frac{1}{\sqrt2}\ket{R_-}\ket{G_2'},\]
    where $G_2'$ is some graph state.
    Now
    \begin{align*}
        0.99999 \leq \abs{\bra{\psi_1}C\ket{\psi_2}}
        =\abs{\bra{\psi_1}U_\pi\ket{\psi_2}}
        &\leq \frac{1}{2}\abs{(\braket{R,R^n}{R,R^n} + \braket{R_-,G_1}{R_-,G_2'}} \\
        &\leq \frac{1}{2} + \frac{1}{2}|\braket{G_1}{G_2'}|.
    \end{align*}
    Thus,
    \[\abs{\braket{G_1}{G_2'}} \geq 0.99999.\]
    Now, since $\ket{G_1},\ket{G_2'}$ are stabilizer states with $\abs{\braket{G_1}{G_2'}} > \frac{1}{\sqrt2}$, it must be the case that $\ket{G_1} = \ket{G_2'}$, and thus, $G_1 = G_2'$. Finally, since $G_1 = G_2' \cong G_2$, we have $G_1 \cong G_2$.
\end{proof}

\subsection{Pure State Clifford Isomorphism Problem with Polynomial Stabilizer Rank States}
\label{sec:low-stab-rank-isomorphism}

In this section, we define a fully classical version of the Pure State Clifford Isomorphism Problem by restricting to polynomial stabilizer rank states, which can be efficiently described and simulated classically. We show that this problem is in $\SZK$, but remains GI-hard.

\subsubsection{$\SZK$ containment}
\label{subsubsec:low-stab-rank-in-szk}
In this section, we present a statistical zero-knowledge protocol for the explicit Clifford nonisomorphism problem when states are given as linear combinations of polynomially many stabilizer states. Importantly, we will use the classical shadows procedure as part of the protocol.

\begin{theorem}
    \label{thm:szk}
    When states are given as a linear combination of polynomially many stabilizer states, the Pure State Clifford Isomorphism Problem is in \textsf{SZK} when $\beta = 1$ and $\alpha < 1$ is a constant.
\end{theorem}
\begin{proof}
    We will describe an \textsf{SZK} protocol for the non-isomorphism problem. One round of the protocol is as follows, where (V) denotes Verifier and (P) denotes Prover. The prover and verifier will repeat the protocol to amplify soundness. Let $\ket{\psi_0},\ket{\psi_1}$ be the input states given as a linear combination of poly$(n)$ stabilizer states.
    \begin{enumerate}
        \item (V) Pick $j \in \{0,1\}$ uniformly at random, and pick a uniformly random Clifford $D \in \mathcal{C}_n$. 
        \item (V) Prepare $N = O(\log(\frac{M}{\gamma \varepsilon^2}))$ classical shadows for the state $\ket{\psi'} = D\ket{\psi_i}$, where $M$ is twice the number of Clifford unitaries on $n$ qubits. Send the description of the classical shadows to (P).
        \item (P) Send $j' \in \{0,1\}$ to (V).
        \item (V) Accept if and only if $j' = j$.
    \end{enumerate}

    First we will prove that if for all Clifford unitaries $C \in \mathcal{C}_n$, $|\bra{\psi_0}C\ket{\psi_1}|< \alpha$,
    then Verifier accepts with probability at least $(1 - \gamma)^2$. 

    Using the classical shadows, the Prover estimates all squared overlaps in
    the Clifford orbit of $\ket{\psi_0}$ and $\ket{\psi_1}$ within additive
    error $\varepsilon$ with probability at least $(1 - \gamma)^2$. For
    $\varepsilon < (1 -\alpha^2)/2$, Prover sends the correct bit, i.e.
    $j' = j$.

    Now suppose there exists a Clifford unitary $D$ such that
    $\abs{\bra{\psi_0}D\ket{\psi_1}}=1$. Then
    \[
        \ketbra{\psi_0}{\psi_0}
        =D\ketbra{\psi_1}{\psi_1}D^\dagger,
    \]
    even though the corresponding state vectors may differ by a global phase.
    Let $\mathcal D_0$ and $\mathcal D_1$ be the distributions of the density
    operator underlying $\ket{\psi'}$ when $j=0$ and $j=1$, respectively. For
    $j\in\{0,1\}$ this density operator is uniform over
    \[
        \left\{
        C\ketbra{\psi_j}{\psi_j}C^\dagger:C\in\mathcal C_n
        \right\}.
    \]
    Right multiplication by $D$ permutes the projective Clifford group, so the
    two orbit distributions are identical: $\mathcal D_0=\mathcal D_1$.
    Classical shadows depend only on these density operators; consequently the
    distribution of the Verifier's message is independent of $j$, and the
    Verifier accepts with probability at most $1/2$.

    Last, we will show that the protocol is zero-knowledge. To simulate Prover's messages for a YES instance, that is, when $\ket{\psi_0}$ and $\ket{\psi_1}$ are not isomorphic, Verifier sends the correct bit with probability $(1-\gamma)^2$. 
\end{proof}

\subsubsection{Graph Isomorphism-hardness}
In this section, we show that the isomorphism problem remains GI-hard even for polynomial stabilizer rank states.

We follow a similar recipe as in \cref{subsec:cip-gi-hardness}, but with a slight modification. Let $G_1, G_2$ be two graphs. We construct the following states:
\begin{align*}
    & \ket{\psi_1} = a \ket{M} + b_1\ket{G_1} \\
    & \ket{\psi_2} = a \ket{M} + b_2\ket{G_2},
\end{align*}
where $\ket{G_1}, \ket{G_2}$ are the graph states for $G_1, G_2$, respectively and $a = 1 - \frac{1}{\Theta(n^{22})}$, $b_1,b_2 = \dfrac{1}{\Theta(n^{11})}$ are chosen such that $\ket{\psi_1},\ket{\psi_2}$ are normalized quantum states. We define $\ket{M}$ over $n$ qubits as the normalized symmetric superposition of all basis states containing exactly two $\ket{R}$ states and $n-2$ computational zero states. 
Then
\begin{align*}
    \ket{M} = c_n \sum_{1 \leq i< j\leq n} \ket{R_i R_j 0 \dots 0},
\end{align*}
where $c_n = \frac{1}{\Theta(n^2)}$ is a normalization constant (note that the elements of the sum are not orthogonal). Because $\ket{R}$ has stabilizer rank $2$, $\ket{M}$ is a linear combination of $\binom{n}{2}$ stabilizer-rank-$4$ states and therefore has stabilizer rank at most $2n(n-1)$.

To prove the theorem we will need some claims and lemmas.

\begin{claim}
    \label{claim:clifford-diagonal}
    Let $C$ be a Clifford unitary such that $C^\dagger Z_i C = Z_{\pi(i)}$ for all $i \in [n]$ and some permutation $\pi \in S_n$. Then $C = U_\pi D$ where $U_\pi$ is the permutation on qubits specified by $\pi$ and $D$ is a diagonal Clifford unitary.
\end{claim}
\begin{proof}
    Let $D = U_\pi^\dagger C$. Then for all $i \in [n]$,
    \begin{align*}
        D^\dagger Z_i D &= C^\dagger (U_\pi Z_i U_\pi ^\dagger) C \\
        &= C^\dagger Z_{\pi^{-1}(i)}C \\
        &= Z_i.
    \end{align*}
    Since, $D^\dagger Z_i D = Z_i$ for all $i$, $D$ commutes with all $Z_i$ and therefore $D$ is diagonal.
\end{proof}

\begin{claim}
    \label{claim:hamming-weight-1}
    Let $\ket{e_k} = \ket{0^{k-1}10^{n-k-1}} \in (\mathbb{C}^2)^{\otimes n}$ be the computational basis state with a one at the $k$-th coordinate of the tensor product and zeros elsewhere. Then 
    \[
        \abs{\braket{e_k}{M}} = \dfrac{1}{\Theta(n)}.
    \]
\end{claim}
\begin{proof}
    Let $\ket{R_{i,j}} = \ket{R_i R_j 0\dots 0}$. 
    For any such state with $k \notin \{i,j\}$, $\braket{e_k}{R_{i,j}} = 0$,
    and for any such state with $k \in \{i,j\}$, 
    \[
        \braket{e_k}{R_{i,j}} = \dfrac{e^{\pi i/8}}{2}.
    \]
    There are $\Theta(n)$ such choices of $\{i,j\}$ such that $k \in \{i,j\}$. Thus
    \[
        \braket{e_k}{M} = c_n \Theta(n)\dfrac{e^{\pi i/8}}{2} = \dfrac{e^{\pi i/8}}{\Theta(n)},
    \]
    so
    \[
        \abs{\braket{e_{k}}{M}} = \dfrac{1}{\Theta(n)}.
    \]
\end{proof}

\begin{claim}
    \label{claim:hamming-weight-2}
    Let
$\ket{e_{k,\ell}} \in (\mathbb C^2)^{\otimes n}$
denote the computational basis state whose bit string has \(1\)'s in positions \(k\) and \(\ell\), and \(0\)'s elsewhere, where \(k \neq \ell\). Then 
    \[
    \abs{\braket{e_{k,\ell}}{M}} = \dfrac{1}{\Theta(n^2)}.
    \]
\end{claim}
\begin{proof}
    Let $\ket{R_{i,j}} = \ket{R_i R_j 0\dots 0}$. 
    For any such state with $k \notin \{i,j\}$, $\braket{e_k}{R_{i,j}} = 0$,
    and for any such state with $k,\ell \in \{i,j\}$, 
    \[
        \braket{e_{k,l}}{R_{i,j}} = \dfrac{e^{\pi i/4}}{2}.
    \]
    There is one such choice of $\{i,j\}$ such that $k,\ell \in \{i,j\}$. Thus
    \[
        \braket{e_{k,\ell}}{M} = c_n \dfrac{e^{\pi i/4}}{2} = \dfrac{e^{\pi i/4}}{\Theta(n^2)},
    \]
    so
    \[
    \abs{\braket{e_{k,\ell}}{M}} = \dfrac{1}{\Theta(n^2)}.
    \]
\end{proof}

\begin{claim}
For every \(n\)-qubit Pauli string
$P$ that contains at least one \(X\) or \(Y\), we have
\[
\abs{\bra M P \ket M}=O\!\left(\frac1n\right).
\]
\label{claim:X-Y-overlap}
\end{claim}

\begin{proof}
Let $\ket{\psi} = \sum_{1 \leq j < k \leq n} |R_j R_k 0\ldots 0\rangle $, then $\ket{M} = \frac{\ket{\psi}}{\|\psi\|}$ and
\[
\bra M P\ket M
=
\frac{\bra \Psi P\ket \Psi}{\|\Psi\|^2}.
\]
Expanding the numerator,
\[
\bra \psi P\ket \psi
=
\sum_{\substack{i_1<j_1\\ i_2<j_2}}
\bra{R_{i_1}R_{j_1}0\ldots 0}
P
\ket{R_{i_2}R_{j_2}0\ldots 0}.
\]
Choose a site \(a\) such that \(P_a\in\{X,Y\}\). If
\[
a\notin \{i_1,j_1,i_2,j_2\},
\]
then the local contribution at site \(a\) is
\[
\bra 0 P_a\ket 0=0.
\]
Hence a summand can be nonzero only if \(a\in\{i_1,j_1,i_2,j_2\}\). There are only \(O(n^3)\) such choices of
\((i_1,j_1,i_2,j_2)\), and each summand has absolute value at most \(1\). Therefore
\[
|\bra \psi P\ket \psi|=O(n^3).
\]
On the other hand,
\[
\|\psi\|^2
=1/c_n^2 = \Theta(n^4).
\]
 Thus
\[
\left|\bra M P\ket M\right|
=
\frac{|\bra \psi P\ket \psi|}{\|\psi\|^2}
=
O\!\left(\frac{n^3}{n^4}\right)
=
O\!\left(\frac1n\right).
\]
\end{proof}

\begin{lem}
    \label{lem:perm2}
    Let $C \in \mathcal{C}_n$ be a Clifford unitary on $n$ qubits. If
    \[\abs{\bra{M}C\ket{M}} \geq 1 - \dfrac{1}{\Theta(n^{11})},\]
    then $C$ is a permutation on qubits up to global phase.
\end{lem}
\begin{proof}
    First note that $\left(1- \dfrac{1}{\Theta(n^{11})}\right)^2 = 1- \dfrac{1}{\Theta(n^{11})}$.
    Let $P \in \mathcal{P}_n$ be a Pauli string. By Hölder's inequality and the Fuchs-van de Graaf inequality,
    \begin{align}
        \begin{split}
            \abs{\bra{M}P\ket{M} - \bra{M}C^\dagger P C\ket{M}} &=
            \abs{\Tr[P(\ketbra{M}{M} - C\ketbra{M}{M}C^\dagger)]} \\
            &\leq \norm{\ketbra{M}{M} - C\ketbra{M}{M}C^\dagger}_1 \norm{P}_\infty \\
            &\leq 2\sqrt{1 - \abs{\bra{M}C\ket{M}}^2} \\
            &\leq 2\sqrt{\frac{1}{\Theta(n^{11})}} < \frac{1}{n^{5}},
        \end{split}
    \end{align}
    for large enough $n$.
    
    Let $Z_i$ be the $Z$ operator on the $i$-th qubit. Its expectation value on $\ket{M}$ is $\bra{M}Z_i\ket{M} = 1 - \frac{1}{\Theta(n^2)}$. Let $P' = C^\dagger Z_i C$. Since $C$ is Clifford, $P' \in \mathcal{P}_n$. By reverse triangle inequality,
    \[ 1 - \frac{1}{\Theta(n^2)} - \frac{1}{n^5} \leq \bra{M}P'\ket{M} \leq 1 - \frac{1}{\Theta(n^2)} + \frac{1}{n^5}. \]
    For large enough $n$, the lower bound forces $P'$ to contain no $X$ or $Y$ operators; see \cref{claim:X-Y-overlap} below. 
    Thus $P' = Z_{\pi(i)}$ for some index $\pi(i)$, meaning $C^\dagger Z_i C = Z_{\pi(i)}$.

    Now, by \cref{claim:clifford-diagonal}, $C = U_\pi D$ where $U_\pi$ is a permutation on qubits and $D$ is a diagonal Clifford. Since $\ket M$ is permutation invariant,
    \[
        \abs{\bra M C\ket M}=\abs{\bra M D\ket M}.
    \]
    Now, up to removing an overall phase,
    \[
        D \ket{x} = f(x) \ket{x},
    \]
    where $f(0)=1$ and $f$ is a quadratic phase polynomial taking values in
    $\{1,i,-1,-i\}$.
    Consequently, if $f$
    is nonconstant, it differs from $f(0)$ on a weight-one or weight-two
    string $x$, and so 
    \begin{align*}
        \abs{\bra M D \ket M} \leq 1 - \abs{\braket{M}{x}}^2 \leq 1 - \dfrac{1}{\Theta(n^4)},
    \end{align*}
    a contradiction.
    Hence $D$ is scalar and
    $C=e^{i\theta}U_\pi$.
\end{proof}

We now proceed with the reduction.

\begin{claim}
    \label{claim:reduction2}
    Let $C \in \mathcal{C}_{n}$ be a Clifford unitary and let
    \begin{align*}
        &\ket{\psi_1} = a\ket{M} + b_1\ket{G_1} \\
        &\ket{\psi_2} = a\ket{M} + b_2\ket{G_2},
    \end{align*}
    where $\ket{G_1}$ and $\ket{G_2}$ are graph states and $a,b_1,b_2$ are real such that $a = 1 - \frac{1}{\Theta(n^{22})}$, $b_1 = \frac{1}{\Theta(n^{11})}$, $b_2 = \frac{1}{\Theta(n^{11})}$, and $\ket{\psi_1}$ and $\ket{\psi_2}$ are valid quantum states. 
    If \[\abs{\bra{\psi_1}C\ket{\psi_2}} \geq 1 - \dfrac{1}{\Omega(n^{22})},\]
    then
    \[\abs{\bra{M}C\ket{M}} \geq 1 - \dfrac{1}{\Omega(n^{11})}\]
    for large enough $n$.
\end{claim}
\begin{proof}
    By the triangle inequality on the expanded inner product, with
    $b=\max\{b_1,b_2\}=1/\Theta(n^{11})$,
    \begin{align*}
        1 - \dfrac{1}{\Omega(n^{22})}
        &\leq \abs{\bra{\psi_1}C\ket{\psi_2}}\\
        &\leq a^2\abs{\bra{M}C\ket{M}}+3b,
    \end{align*}
    Rearranging yields
    \[
        \abs{\bra M C\ket M}
        \geq\frac{1-1/\Omega(n^{22})-3b}{a^2}
        =1-\frac{1}{\Omega(n^{11})},
    \]
    as desired.
\end{proof}

Combining \cref{claim:reduction2} and \cref{lem:perm2}, if there exists a Clifford unitary $C \in \mathcal{C}_n$ such that $\abs{\bra{\psi_1}C\ket{\psi_2}} \geq 1 - \frac{1}{\Omega(n^{22})}$, then $C$ must be a permutation up to global phase.

\begin{claim}
    \label{claim:reduction-ab-relationship}
    Let $a,b_1,b_2$ be defined as in \cref{claim:reduction2}.
    Then 
    \[
        \dfrac{1- a^2}{b_1b_2} \geq 1 - o(1).
    \]
\end{claim}
\begin{proof}
    First, we can see that $a^2 + b_1^2 \approx 1$ and $a^2 + b_2^2 \approx 1$ by showing that $\ket M$ and $\ket G$ are approximately orthogonal for any graph state $\ket{G}$.
    In particular, since any graph state is equally supported on all $2^n$ computational basis states and $\ket{M}$ is supported on $\Theta(n^2)$ computational basis states, 
    $\abs{\braket{M}{G}}^2$ is exponentially small in $n$.
    Thus, 
    \begin{align*}
        & a^2 + b_1^2 = 1 + \eps_1, \\
        & a^2 + b_2^2 = 1 + \eps_2,
    \end{align*}
    where $|\eps_1| \leq \exp(-n)$ and $|\eps_2| \leq \exp(-n)$.
    WLOG let $b_1 \geq b_2$. Now 
    \begin{align*}
        \dfrac{1 - a^2}{b_1b_2} \geq \dfrac{1 - a^2}{b_1^2}  = \dfrac{b_1^2 - \eps_1}{b_1^2} = 1 - o(1),
    \end{align*}
    as desired.
\end{proof}

\begin{claim}
    Let $\ket{\psi_1}, \ket{\psi_2}$ be defined as above. If there exists a Clifford unitary $C$ such that
    \[\abs{\bra{\psi_1}C\ket{\psi_2}} \geq 1 - \dfrac{1}{\Omega(n^{23})},\]
    then $G_1 \cong G_2$.
\end{claim}
\begin{proof}
    By \cref{lem:perm2} and \cref{claim:reduction2},
    $C=e^{i\theta}U_\pi$ for a qubit permutation $U_\pi$. The scalar phase
    disappears from the absolute value. Moreover,
    \[U_\pi\ket{\psi_2}=a\ket M+b_2\ket{G_2'},\]
    where $\ket{G_2'}=U_\pi\ket{G_2}$ is the graph state for some graph $G_2'$.
    Expanding yields
    \begin{align*}
        1 - \dfrac{1}{\Omega(n^{23})}
        \leq \abs{\bra{\psi_1}C\ket{\psi_2}}
        =\abs{\bra{\psi_1}U_\pi\ket{\psi_2}}
        &= \abs{a^2\braket{M}{M} + b_1b_2\braket{G_1}{G_2'} + ab_2\braket{M}{G_2'} + ab_1\braket{G_1}{M}} \\
        &\leq a^2 + b_1b_2\abs{\braket{G_1}{G_2'}} + (ab_2 + ab_1)s,
    \end{align*}
    where $s$ is the maximum magnitude of the overlap of $\ket M$ with any
    graph state. The state $\ket M$ is supported on only $O(n^2)$ computational
    basis vectors, whereas every graph-state amplitude has magnitude $2^{-n/2}$.
    Hence, by Cauchy--Schwarz,
    $s\leq \sqrt{O(n^2)}2^{-n/2}=\exp(-\Omega(n))$. Rearranging and absorbing
    the exponentially small $s$ terms yields
    \begin{align*}
        \abs{\braket{G_1}{G'_2}} &\geq \dfrac{1}{b_1b_2} \left(1 - \dfrac{1}{\Omega(n^{23})} - a^2 \right) \\
        &\geq \dfrac{1-a^2}{b_1b_2} - \dfrac{1}{\Omega(n)} \\
        &\geq 1 - o(1),
    \end{align*}
    where we used the fact that $b_1b_2 = \dfrac{1}{\Theta(n^{22})}$ and \cref{claim:reduction-ab-relationship}. Finally, 
    because $\ket{G_1}$ and $\ket{G_2'}$ are stabilizer states with $\abs{\braket{G_1}{G'_2}} \geq 1 - o(1)$, we have $\ket{G_1} = \ket{G_2'}$ and so $G_1 = G_2'$. Since $G_1 = G_2'$ and $G_2' \cong G_2$, we have $G_1 \cong G_2$, as desired.
\end{proof}

\section{Mixed State Group Isomorphism Problems}
\label{sec:mixedstates}
In this section, we study the complexity of Mixed State Group Isomorphism Problems. \cite{lockhart2017quantum} showed that for the symmetric group, this problem is $\QSZK$-hard, but left open whether one can prove tight upper bounds for this problem. We show that not only is this problem $\QSZK$-complete, in fact the Mixed State Group Isomorphism Problem is $\QSZK$-complete for \textit{all} finite groups.

\subsection{$\QSZK$ containment of Mixed State Group Isomorphism}
In this section, we prove that Mixed State Group Isomorphism is in $\QSZK$. We will reduce to the well-known $\QSZK$-complete problem Quantum State Distinguishability \cite{watrous2002quantum}.

We will also use the following claim.
\begin{claim}
    \label{claim:twirls}
    Let $\rho,\sigma$ be two density matrices, and
    let $\mathcal{S}$ be a set of unitaries such that 
    \begin{align*}
        \forall U,V \in \mathcal{S}, F(V\rho V^\dagger,U\sigma U^\dagger) \leq \eps.
    \end{align*}
    Furthermore let $\mathcal{E}$ be the twirling channel for $\mathcal{S}$. That is,
    \begin{align*}
        \mathcal{E}(\rho) = \frac{1}{\abs{\mathcal{S}}}\sum_{U\in \mathcal{S}} U \rho U^\dagger.
    \end{align*}
    Then 
    \begin{align*}
        F(\mathcal{E}(\rho),\mathcal{E}(\sigma)) \leq \eps \abs{\mathcal{S}}.
    \end{align*}
\end{claim}

\begin{proof}
    First note the following fact due to Rotfel'd \cite{rotfel1969singular}. Let $A$ and $B$ be PSD square matrices. Then 
    \begin{align*}
        \Tr\sqrt{A+B} \leq \Tr\sqrt{A} + \Tr\sqrt{B}.
    \end{align*}
    Now, recall the definition of the fidelity:
    \begin{align*}
        F(\rho,\sigma) = \Tr\sqrt{\rho^{1/2} \sigma \rho^{1/2}}.
    \end{align*}
    Let $\rho_j = U_j\rho U_j^\dagger$, where $U_j$ is the $j$-th element of $\mathcal{S}$ for some fixed ordering, and similarly let $\sigma_\ell = U_\ell \sigma U_\ell^\dagger$. Then
    \begin{align*}
        \mathcal{E}(\rho) &= \dfrac{1}{\abs{\mathcal{S}}}\sum_j \rho_j, \\
        \mathcal{E}(\sigma) &= \dfrac{1}{\abs{\mathcal{S}}}\sum_\ell \sigma_\ell.
    \end{align*}
    Now
    \begin{align}
    \begin{split}
        F(\mathcal{E}(\rho),\mathcal{E}(\sigma)) =& \Tr\sqrt{(\dfrac{1}{\abs{\mathcal{S}}}\sum_j \rho_j)^{1/2}(\dfrac{1}{\abs{\mathcal{S}}}\sum_\ell \sigma_\ell)(\dfrac{1}{\abs{\mathcal{S}}}\sum_j \rho_j)^{1/2}} \\
        =& \Tr\sqrt{\dfrac{1}{\abs{\mathcal{S}}}\sum_{\ell}(\dfrac{1}{\abs{\mathcal{S}}}\sum_j \rho_j)^{1/2}\sigma_\ell(\dfrac{1}{\abs{\mathcal{S}}}\sum_j \rho_j)^{1/2}} \\
        \leq& \dfrac{1}{\sqrt{\abs{\mathcal{S}}}}\sum_{\ell} \Tr\sqrt{(\dfrac{1}{\abs{\mathcal{S}}}\sum_j \rho_j)^{1/2}\sigma_\ell(\dfrac{1}{\abs{\mathcal{S}}}\sum_j \rho_j)^{1/2}} \\
        =& \dfrac{1}{\sqrt{\abs{\mathcal{S}}}}\sum_{\ell} \Tr\sqrt{\sigma_\ell^{1/2}(\dfrac{1}{\abs{\mathcal{S}}}\sum_j \rho_j)\sigma_\ell^{1/2}} \\ 
        =& \dfrac{1}{\sqrt{\abs{\mathcal{S}}}}\sum_{\ell} \Tr\sqrt{\dfrac{1}{\abs{\mathcal{S}}}\sum_j \sigma_\ell^{1/2}\rho_j\sigma_\ell^{1/2}} \\
        \leq& \dfrac{1}{\abs{\mathcal{S}}}\sum_{\ell,j} \Tr\sqrt{\sigma_\ell^{1/2}\rho_j\sigma_\ell^{1/2}} \\ 
        \leq& \abs{\mathcal{S}} \max_{j,l} F(\rho_j,\sigma_l),
    \end{split}
    \end{align}
    where we used the linearity of matrix multiplication, the above trace inequality, and the symmetry of fidelity in its inputs.
\end{proof}

\begin{theorem} [Restatement of the upper bound in \cref{thm:mixed-qszk-complete}]
    Let $G$ be any nontrivial finite group, let $\alpha < 1$ be a constant, and $\beta\geq 1-1/\omega(\log|G|)$. Then $(\alpha,\beta)-\mathrm{MSGI}[G]$ is in \textsf{QSZK}.
\end{theorem}

\begin{proof}
    For convenience, we prove the theorem for the complement problem, Mixed State Group Nonisomorphism, which suffices because $\QSZK$ is closed under complement.
    We reduce to the \textsf{QSZK}-complete problem Quantum State Distinguishability (QSD) \cite{watrous2002quantum}. The reduction is as follows:
    Let $\rho$ and $\sigma$ be the input states to the Mixed State Group Nonisomorphism problem.
    We define the inputs to the QSD instance as
    \begin{align*}
        \rho' =& \frac{1}{\abs{G}}\sum_{g \in G} (R(g) \rho R(g)^\dagger)^{\otimes k}, \\
        \sigma' =& \frac{1}{\abs{G}}\sum_{g \in G} (R(g) \sigma R(g)^\dagger)^{\otimes k},
    \end{align*}
    where $k \geq \omega(\log|G|)$ and $k(1-\beta) \leq 1/20$.
    Let's first consider as input to the reduction a YES instance of Mixed State Clifford Nonisomorphism. We are given circuits describing $n$-qubit states $\rho, \sigma$ such that for all $g \in G$, 
    \begin{align*}
        F(\rho, R(g) \sigma R(g)^\dagger) \leq \alpha.
    \end{align*}
    Applying \cref{claim:twirls} with states $\rho^{\otimes k}$, $\sigma^{\otimes k}$ and $\mathcal{S} = \{R(g)^{\otimes k} \mid g \in G\}$ yields
    \begin{align*}
        F(\rho',\sigma') \leq \alpha^k\abs{G} \leq o(1).
    \end{align*}
    We also used the fact that $F(\rho^{\otimes k }, \sigma^{\otimes k}) = F(\rho,\sigma)^k$. Importantly, the states $\rho'$ and $\sigma'$ are efficiently preparable by sampling a random group element and applying it to $k$ copies of $\rho$ and $\sigma$, respectively.
    Now, since the fidelity is negligible, by the Fuchs-van de Graaf inequalities, the trace distance is exponentially close to $1$:
    \begin{align*}
        \dfrac12\norm{\rho' - \sigma'}_1 \geq 1 - o(1).
    \end{align*}

    Now, let's consider a NO instance of Mixed State Nonisomorphism, where we are given $\rho, \sigma$ such that for some $h \in G$,
    \begin{align*}
        F(\rho, R(h) \sigma R(h)^\dagger) \geq \beta.
    \end{align*}
    Let $\mathcal{E}$ be the $k$-twirling channel for $G$, that is 
    \begin{align*}
        \mathcal{E}(\rho^{\otimes k}) = \dfrac{1}{\abs{G}}\sum_{g \in G} (R(g) \rho R(g)^{\dagger})^{\otimes k}.
    \end{align*}
    Then 
    \begin{align*}
        F(\rho', \sigma') =& F(\mathcal{E}(\rho ^{\otimes k}), \mathcal{E}(\sigma^{\otimes k})) \\ 
        =& F(\mathcal{E}(\rho^{\otimes k}), \mathcal{E}((R(h)\sigma R(h)^{\dagger})^{\otimes k})) \\
        \geq& F(\rho^{\otimes k}, (R(h)\sigma R(h)^{\dagger})^{\otimes k}) \\
        =& F(\rho, R(h)\sigma R(h)^{\dagger})^k \\
        =& \beta^k \\ 
        \geq& 1 - k(1-\beta) \\
        \geq& 0.95,
    \end{align*}
    where we used multiplicativity and monotonicity of fidelity, Bernoulli's inequality, and our assumption that $k(1-\beta) \leq \dfrac{1}{20}$. Applying Fuchs-van de Graaf once more yields 
    \begin{align*}
        \dfrac12\norm{\rho' - \sigma'}_1 \leq 1/3.
    \end{align*}
    Thus, our reduction maps YES instances of Mixed State Nonisomorphism to YES instances of QSD, and NO instances of Mixed State Nonisomorphism to NO instances of QSD.
\end{proof}

\subsection{Mixed State Group Isomorphism Problem is $\QSZK$-hard}
\label{subsec:mixed-qszk-complete}

In this section, we prove that the Mixed State Group Isomorphism Problem is \textsf{QSZK}-hard over any finite group $G$.

We will also use the following facts:
\begin{theorem}[Efficient construction of approximate state $t$-designs \cite{brandao2016local}]
    \label{fact:designs}
    Local random circuits of depth $\mathrm{poly}(n,t,\log \frac1\eps)$ form $\eps$-approximate $t$-designs.
\end{theorem}

We will also use the following claims:

\begin{claim}[Fidelity concentration inequality for approximate state $t$-designs]
    \label{claim:maxfidelity}
    Let $\rho$ be a fixed $n$-qubit mixed quantum state.
    Let $G$ be a group.
    Let $\ket\psi$ be a pure state drawn from an $\eps$-approximate state $t$-design where $t = \omega(\max\{\log |G|, n\})$ and $\eps \leq \frac{1}{2^{nt}}$. Then with probability $1 - o(1)$, 
    \begin{align*}
        \max_{g \in G} F(\rho, R(g)\ket\psi) \leq o(1).
    \end{align*}
\end{claim}
\begin{proof}
    Applying Markov's inequality and the moment method, we get
    \begin{align*}
        \Pr_{\ket{\psi}\sim \mu}[F(\rho, R(g)\ket\psi)^{2} \geq 1/\sqrt{d}] &= \Pr_{\ket{\psi}\sim \mu}[F(\rho, R(g)\ket\psi)^{2t}
        \geq 1/d^{t/2}] \\ 
        &\leq d^{t/2}\Ex[F(\rho, R(g)\ket\psi)^{2t}] \\
        &\leq \dfrac{d^{t/2}}{d^t} + \eps d^{t/2} \leq \dfrac{2}{d^{t/2}},
    \end{align*}
    where $d = 2^n$.
    We also used the fact that $\mu$ is an approximate $2t$ design and a standard bound on the moments of the fidelity for Haar random states.
    Since $t \geq \omega(\log|G|)$, and by a union bound over all group elements,
    \begin{align*}
        \Pr_{\ket\psi \sim \mu}[\max_{g \in G}F(\rho, R(g)\ket\psi)^{2} \geq 1/\sqrt{d}] \leq \dfrac{\abs{G}}{2^{\omega(\log|G|)}} = o(1),
    \end{align*}
    and we are done.
\end{proof}

\begin{claim}
    \label{claim:selffidelity}
    Let $S$ be a set of unitaries corresponding to the image of a group $G$ under a representation $R$.
    Let $\ket\psi$ be a pure state drawn from an $\eps$-approximate state $t$-design where $t = \omega(\max\{\log |G|, n\})$ and $\eps \leq \frac{1}{2^{nt}}$. Then with probability $1 - o(1)$, 
    \begin{align*}
        \max_{g \in G \setminus \{I\}} F(R(g)\ket\psi, \ket\psi) \leq \max_{g \in G \setminus \{I\}} \left( \frac{|\Tr(R(g))|}{2^n} \right) + o(1).
    \end{align*}
\end{claim}

\begin{proof}
    For all $g \in G$, define the complex random variable $X_g = \langle \psi | R(g) | \psi \rangle = \Tr(R(g) \ketbra{\psi}{\psi})$. Since $\ket\psi$ is drawn from a $t$-design with $t \geq 1$, its first moment is given by the Haar measure:
    \begin{align*}
        \mathbb{E}_{\ket\psi}[X_g] &= \Tr\left(R(g) \mathbb{E}_{\ket\psi}[\ketbra{\psi}{\psi}]\right) \\
        &= \Tr\left(R(g) \frac{I}{2^n}\right) \\
        &= \frac{\Tr(R(g))}{2^n}.
    \end{align*}
    Consequently, the absolute value of the expectation is:
    \begin{align*}
        |\mathbb{E}_{\ket\psi}[X_g]| = \frac{|\Tr(R(g))|}{2^n}.
    \end{align*}
    By concentration of measure for $\eps$-approximate $t$-designs (see for example \cref{claim:maxfidelity} or \cite{mele2024introduction}), $X_g$ deviates from its mean under the Haar measure by more than some $o(1)$ with probability at most $\frac{1}{d^{\Omega(t)}}$. Thus by a union bound, with high probability, the maximum deviation over the group is bounded by
    \begin{equation*}
        \max_{g \in G \setminus \{I\}} |X_g - \mathbb{E}_{\ket\psi}[X_g]| \leq o(1).
    \end{equation*}
    Next we apply the triangle inequality $|X_g| \leq |\mathbb{E}_{\ket\psi}[X_g]| + |X_g - \mathbb{E}_{\ket\psi}[X_g]|$ to obtain, with high probability:
    \begin{equation*}
        \max_{g \in G \setminus \{I\}} F(R(g)\ket\psi, \ket\psi) \leq \max_{g \in G \setminus \{I\}} \left( \frac{|\Tr(R(g))|}{2^n} \right) + o(1).
    \end{equation*}
\end{proof}

We now prove the theorem.

\begin{theorem} [Formal restatement of the lower bound in \cref{thm:mixed-qszk-complete}]
    Let $G$ be a finite group with $\log|G| \leq \mathrm{poly}(n)$, and let $R$ be a unitary representation of $G$ such that $\max_{g \neq I}\dfrac{|\Tr(R(g))|}{2^n} = \mu < 1$. Then $(\alpha,\beta)$-MSGI$[G]$ is \textsf{QSZK}-hard under randomized reductions for any $\alpha > \dfrac{1}{2} + \dfrac{\mu}{2}$ and $\beta > \alpha$. 
\end{theorem}
\begin{proof}
    We will proceed by reducing from $(\alpha',\beta')$-Quantum State Distinguishability, which importantly remains $\QSZK$-hard even for $\alpha' = \exp(-n)$, $\beta' = 1 - \exp(-n)$.
    Given states $\rho,\sigma$, we define the inputs to the Mixed State Group Isomorphism instance as 
    \begin{align*}
        \sigma_0 =& \dfrac12 \rho + \dfrac12 \ketbra{\psi}{\psi} \\
        \sigma_1 =& \dfrac12 \sigma + \dfrac12 \ketbra{\psi}{\psi},
    \end{align*}
    where $\ket{\psi}$ is drawn from an $\eps$-approximate $t$-design where  $t = \omega(\max\{\log |G|, n\})$ and $\eps \leq \frac{1}{2^{nt}}$. Importantly, preparing the descriptions of the states is efficient since sampling random circuits according to (\cref{fact:designs}) is efficient.

    Now, suppose $\rho,\sigma$ corresponds to a NO instance of $(\alpha',\beta')$-QSD. Taking the isomorphism to be the identity, $h = I$, we have
    \[  
        \norm{\left(\dfrac12 \rho + \dfrac12 \ketbra{\psi}{\psi} \right) - \left(\dfrac12 \sigma + \dfrac12\ketbra{\psi}{\psi}\right)}_1 = \dfrac12\norm{\rho - \sigma}_1 \leq \alpha'.
    \]
    Thus, by the Fuchs-van de Graaf inequalities, 

    \[
        F(\sigma_0,\sigma_1) \geq 1 - \alpha' = 1- o(1).
    \]

    On the other hand, suppose $\rho,\sigma$ corresponds to a YES instance of $(\alpha',\beta')$-QSD. Let's analyze the fidelity under an arbitrary unitary $U \in G$ with $U \neq I$. Using the Rotfel'd trace inequality and by \cref{claim:maxfidelity} and \cref{claim:selffidelity}, with probability $1 - o(1)$,
    \begin{align*}
        F(U\sigma_0 U^\dagger,\sigma_1) \leq& \dfrac12F(U\rho U^\dagger, \sigma) + \dfrac12F(U\rho U^\dagger, \ketbra{\psi}{\psi}) + \dfrac12F(U\ketbra{\psi}{\psi}U^\dagger, \sigma) + \dfrac12F(U\ketbra{\psi}{\psi}U^\dagger, \ketbra{\psi}{\psi}) \\
        \leq& \dfrac12F(U\rho U^\dagger, \sigma) + \dfrac12F(U\ketbra{\psi}{\psi}U^\dagger, \ketbra{\psi}{\psi}) +  o(1) \\
        \leq& \dfrac12 + \dfrac{\mu}{2} +  o(1).
    \end{align*}
    On the other hand, for $U = I$ we have 
    \begin{align*}
        F(\sigma_0,\sigma_1) \leq \dfrac12F(\rho,\sigma) + F(\rho,\ketbra{\psi}) + \dfrac12F(\ketbra{\psi},\ketbra{\psi}) \leq \dfrac12 + o(1).
    \end{align*}
 \end{proof}

One might ask whether the condition $\max_{g \in G}\dfrac{|\Tr(R(g))|}{2^n} = \mu < 1$ is natural, as well as whether it is necessary. Importantly, one can choose pathological low-dimensional representations for which the condition is not satisfied, such as the one-dimensional representation of a cyclic group embedded in a single diagonal entry of a larger unitary. However, we note that for natural representations appearing in the literature \cite{bouland2025state, hinsche2025abelian}, the trace condition is met.

\subsection{Mixed StateHSP is $\QSZK$-hard}
\label{subsec:mixedHSP-qszk-hard}

Here we prove that the mixed state version of StateHSP is $\QSZK$-hard for some abelian groups. In particular, finding a hidden subgroup that stabilizes a mixed state is hard when the group $G$ contains an element of order $2$. Importantly, this captures groups for which the pure StateHSP problems are in $\BQP$, such as the group and representation of $\mathbb{Z}_2^n$ given in \cite{bouland2025state}. The $\QSZK$-hardness also holds for the decision problem, which asks whether the hidden subgroup is trivial. We now present the theorem.

\begin{theorem}
    Let $G$ be a finite abelian group such that there exists an element $g \in G$ with order $2$. Then Mixed StateHSP over $G$ is \textsf{QSZK}-hard.
\end{theorem}
\begin{proof}
    We reduce from $(\alpha,\beta)$-Quantum State Distinguishability. Importantly QSD remains $\QSZK$-hard even when $\alpha = 1/\exp(n)$ and $\beta = 1 - 1/\exp(n)$.
    
    Let $\sigma_1,\sigma_2$ be the input states to the QSD instance.
    Let $h \in G$ be a group element of order $2$, and let 
    $\ket{v_1}, \ket{v_2}$ be orthonormal vectors such that $R(h)\ket{v_1} = \ket{v_2}$ and $R(h)\ket{v_2} = \ket{v_1}$. We call these the \textit{label} vectors. Importantly, such vectors exist because $R(h)^2 = I$ implies that $R(h)$ is Hermitian with eigenvalues $\pm1$, and therefore has orthonormal eigenvectors $\ket{h_+},\ket{h_-}$. Thus, we can take 
    \begin{align*}
        \ket{v_1} &= \frac{1}{\sqrt2}\left( \ket{h_+} + \ket{h_-} \right) \\
        \ket{v_2} &= \frac{1}{\sqrt2}\left( \ket{h_+} - \ket{h_-} \right) 
    \end{align*}
    to satisfy
    \begin{align*}
        R(h)\ket{v_1} &= \frac{1}{\sqrt2}\left( \ket{h_+} - \ket{h_-} \right) = \ket{v_2} \\
        R(h)\ket{v_2} &= \frac{1}{\sqrt2}\left( \ket{h_+} + \ket{h_-} \right) = \ket{v_1}.
    \end{align*}
    
    We take as input to the mixed StateHSP instance the state
    \begin{align*}
        \rho = \frac12\ketbra{v_1}{v_1} \otimes \sigma_1 + \frac12\ketbra{v_2}{v_2} \otimes \sigma_2,
    \end{align*}
    which is efficiently preparable.
    We define the representation $R'$ as
    \begin{align*}
        R'(g) = R(g) \otimes I,
    \end{align*}
    where $R(g)$ acts on the label register.

    Now, note that
    \begin{align*}
        \norm{\rho - R(h)\rho R(h)^\dagger}_1 &= \norm{\frac12\ketbra{v_1}{v_1} \otimes \sigma_1 + \frac12\ketbra{v_2}{v_2} \otimes \sigma_2 - \frac12\ketbra{v_2}{v_2} \otimes \sigma_1 - \frac12\ketbra{v_1}{v_1} \otimes \sigma_2}_1 \\
        &= \dfrac12 \norm{\ketbra{v_1}{v_1} \otimes \left( \sigma_1 - \sigma_2 \right)}_1 + \dfrac12 \norm{\ketbra{v_2}{v_2} \otimes \left( \sigma_1 - \sigma_2 \right)}_1 \\
        &= \norm{ \sigma_1 - \sigma_2}_1.
    \end{align*}
    Thus, $\rho$ is $(1-1/\exp(n))$-stabilized by $R(h)$ if and only if $\sigma_1,\sigma_2$ are exponentially close in trace distance. Therefore, it suffices to check whether $h$ is in the hidden subgroup $H$ of $\rho$. Fortunately, this is an instance of the group membership problem, which is efficient for abelian groups \cite{mckenzie1987parallel}. Thus, if we can solve mixed StateHSP, we can determine whether $\sigma_1$ and $\sigma_2$ are close or far in trace distance.
\end{proof}

\section{Bosonic Group Isomorphism Problems}
\label{sec:bosons}

We define a bosonic version of the state isomorphism problem, based on Gaussian and linear optical quantum elements. Formally, we define

\begin{definition}
[$(\alpha,\beta; r)$-Pure State Linear Optical Isomorphism Problem]
Given two core states $\ket{c_1}, \ket{c_2}$ over $n$ modes with $r=O(1)$ photons, and precision parameters $(\alpha,\beta)$  the problem of $(\alpha,\beta; r)$-PSGI$[\mathcal O_n]$ is to decide which one of the following statements holds:
\begin{itemize}
    \item There exists a linear optical unitary $V\in\mathcal O_n$ such that $\abs{\bra{c_2}V\ket{c_1}} \geq \beta$.
    \item For all $V\in\mathcal O_n$ we have $\abs{\bra{c_2}V\ket{c_1}} \le \alpha$.
\end{itemize}
\end{definition}

Below we show a Graph Isomorphism lower bound for this problem, with $r=3$.

\subsection{Graph Isomorphism lower bound}

\begin{theorem} [Formal restatement of the lower bound in \cref{thm:bosonic-PSGI*}]\label{thm:cv-gi}
Let $\alpha \geq 1-\frac{1}{96 n^5}$ and $\beta > \alpha$. Then
\begin{align}
\text{\textsc{GI}}\leq_{\mathsf{P}} (\alpha,\beta; 3)\text{-\textsc{PSGI}}[\mathcal O_n],
\end{align}
\end{theorem}
The idea of the proof is as follows: We encode an undirected graph $G = (V, E)$ into a bosonic state with three photons with the following polynomial
\begin{align}\label{eq:encode-graph-into-poly}
F_{\mathcal G}(\mathbf z) = \sum_i z_i^3 + \sum_{i<j} A_{ij} z_i z_j.
\end{align}
For two graphs $\mathcal G_1, \mathcal G_2$, one can show there exists a unitary $U\in\mathrm{U}(n)$ such that $F_{\mathcal G_1}(\mathbf z) = F_{\mathcal G_2}(U\mathbf z)$ if and only if $\mathcal G_1$ and $\mathcal G_2$ are isomorphic. The reason is that the cubic term $\sum_i z_i^3$ enforces the unitary $U$ to be a permutation up to phases. This is the same observation that the only linear transformations that preserve $\ell_p$ norm are permutations for any $p\neq 1,2$. This was also pointed out by Aaronson in \cite{aaronson2004quantum}. Below is a quick proof of this fact for the sake of completeness.

\begin{lem}\label{lem:aux}
Let
\begin{align}
f(\mathbf z) = \sum_i z_i^3.
\end{align}
Then, for any unitary $T$ such that $f(T\mathbf z) = f(\mathbf z)$, we have that
\begin{align}
T = PD
\end{align}
for some phase diagonal matrix $D = \mathrm{diag}(e^{i\phi_1}, \cdots, e^{i\phi_n})$ and some permutation matrix $P$.
\end{lem}
\begin{proof}
Let $g(\mathbf z) = f(T\mathbf z)$, and consider the Hessian matrices of $g$ and $f$. We have
\begin{align}
\begin{split}
\boldsymbol\nabla^2 f &= \left(\frac{\partial^2}{\partial z_i \partial z_j} f\right)_{ij} = 6\, \mathrm{diag}(z_1,\cdots, z_n),\\
\boldsymbol\nabla^2 g &= 6 \left( \sum_k T_{ki} T_{kj} (T\mathbf z)_{k}\right) = 6 \, T^t \mathrm{diag}\left((T\mathbf z)_1, \cdots, (T\mathbf z)_n  \right) T.
\end{split}
\end{align}
Imposing $f=g$ gives
\begin{align}
\begin{pmatrix}
z_1 & \cdots & 0\\
 & \ddots &\\
 0 & \cdots &z_n
\end{pmatrix}
=
T^t
\begin{pmatrix}
(T\mathbf z)_1 & \cdots & 0\\
 & \ddots &\\
 0 & \cdots &(T\mathbf z)_n
\end{pmatrix}
T,
\end{align}
for all $\mathbf z \in \mathbb C^n$. Note that this equation tells us that
\begin{align}
T^t
\begin{pmatrix}
(T\mathbf z)_1 & \cdots & 0\\
 & \ddots &\\
 0 & \cdots &(T\mathbf z)_n
\end{pmatrix}
T
\end{align}
is diagonal for all $\mathbf z \in \mathbb C^n$. Since $T$ is invertible, we have that $T\mathbf z$ covers all of $\mathbb C^n$. This yields that $T^t \Lambda T$ is diagonal for all diagonal matrices $\Lambda$. In other words, for all $\lambda_i\in\mathbb C$ and any $i\neq j$ we have:
\begin{align}
\sum_{k} \lambda_{k} T_{ik} T_{jk} = 0.
\end{align}
Choosing $(\lambda_k)_k$ to be non-zero only at some $k$ we get
\begin{align}
T_{ik} T_{jk} = 0 \quad \forall j\neq i, \forall k
\end{align}
This yields that $T$ has at most one non-zero entry in each column. Therefore, we can write it as
\begin{align}
T = PD
\end{align}
for $D$ being the elements in each column, and $P$ being a permutation matrix (which essentially specifies the non-zero entry in every column). Finally, we get that $D$ must be only phases due to the unitarity of $T$.
\end{proof}

\iffalse
We now proceed with the rest of our proof. Note that we can write
\begin{align}
P_{\mathcal G}(\mathbf z) = C(\mathbf z) + Q_{\mathcal G}(\mathbf z),
\end{align}
where
\begin{align}
C(\mathbf z) = \sum_i z_i^3, \quad Q_{\mathcal G}(\mathbf z) = \sum_{i,j} A_{ij} z_i z_j. 
\end{align}
Note that a unitary transformation preserves the degree of our polynomial. In other words, the transformation $U$ must map $C$ to $C$ and $Q_{\mathcal G_1}$ to $Q_{\mathcal G_2}$. From the transformation between cubic parts, and \cref{lem:aux} we get that
\begin{align}
U = PD.
\end{align}
Denoting the adjacency matrices of $\mathcal G_1, \mathcal G_2$ by $A_1, A_2$, and plugging the transformation into the quadratic part, we get that
\begin{align}\label{eq:A-transformation}
\begin{split}
U^T A_1 U &= A_2,\\
\therefore D(P^tA_1P) D &= A_2.
\end{split}
\end{align}
Finally, we note:
\begin{lem}\label{lem:DBD}
For a matrix $B\in\mathbb R_{\geq 0}^{n\times n}$ and a diagonal phase matrix $D = \mathrm{diag}(e^{i\phi_1}, \cdots, e^{i\phi_n})$, if $DBD\in\mathbb R^{n\times n}_{\geq 0}$, then $B = DBD$.
\end{lem}
\begin{proof}
Note that
\begin{align}
\bra{k} DBD\ket{j} = e^{i\phi_j+\phi_k} B_{kj},
\end{align}
and from the non-negativity of entries of $DBD$ we have
\begin{align}
\bra{k} DBD\ket{j} = |e^{i\phi_j+\phi_i} B_{kj}| = B_{kj}.
\end{align}
\end{proof}
Finally, putting \eqref{eq:A-transformation} together with \cref{lem:DBD} we get that
\begin{align}
P^tA_1P = A_2,
\end{align}
which is saying that $\mathcal G_1$ and $\mathcal G_2$ are isomorphic graphs.
\fi

We now proceed to the proof of \cref{thm:cv-gi}.

\begin{proof}[Proof of \cref{thm:cv-gi}]
We consider the encoding presented in \cref{eq:encode-graph-into-poly} with the following normalization
\begin{align}
P_{\mathcal G}(\mathbf z) = \frac{1}{\sqrt{12 n}} \sum_{i} z_i^3 + \frac{1}{\sqrt{2|E|}} \sum_{i<j} A_{ij} z_i z_j.
\end{align}
First, note that if the two graphs are isomorphic, then the equivalence between the states holds trivially. We need to prove the other direction: that if the overlap exceeds $1-O(\frac1{n^5})$, then the two graphs are isomorphic. Without loss of generality, we assume that the given graphs have the same number of edges. If this condition is not satisfied, we can simply reject (e.g., map to the isomorphism of $\ket{0}$ and $\ket{1}$).

We first show that if the overlap is larger than $1-\varepsilon = 1-O(\frac1{n^5})$, then from the cubic part's transformation we get $\norm{U - PD} \leq \varepsilon$ for some diagonal phase matrix $D$ and some permutation $P$. Then, using the fact that the distance between any two distinct binary matrices is at least $1$, we show that the transformation $U$ must be a permutation. The following lemma establishes a bound between $U$ and some $PD$.

\begin{lem}\label{lem:helper-gapped-cv}
Let $U$ be an $n\times n$ unitary. For any $0<\delta<\frac1n \left( \frac{3-\sqrt 5}{2} \right)\approx 0.38/n$, we have that if 
\begin{align}
\abs{\left\langle\frac{1}{\sqrt{6n}}C(U\mathbf z) |\frac{1}{\sqrt{6n}}C(\mathbf z) \right\rangle} \geq 1 -\delta,
\end{align}
then 
\begin{align}
\norm{U - PD} \leq \sqrt{3n \delta}
\end{align}
for some diagonal phase matrix $D$ and some permutation $P$.
\end{lem}
\begin{proof}
Let $\mathbf z' = U \mathbf z$. Hence, $z'_i = \sum_{j} U_{ij} z_j$. Then
\begin{align}
\langle z_k^3 | z_i'{}^3\rangle = U_{ik}^3 \norm{z_k^3}^2 = 6 U_{ik}^3.
\end{align}
As a result, we get
\begin{align}
\frac{1}{6n}\bra{C(U \mathbf z)} C(\mathbf z)\rangle = \frac{1}{n} \sum_{i,k=1}^n  U_{ik}^3.
\end{align}
From the bound on the inner product, we obtain
\begin{align}\label{eq:U_ik}
1-\delta \leq \frac1n \sum_{i,k} |U_{ik}|^3 \leq \frac1n \sum_i \left( \sum_{k} |U_{ik}|^3 \right) \leq \frac1n \sum_i \max_{k} |U_{ik}| \cdot \sum_{k} |U_{ik}|^2 =  \frac1n \sum_{i} \max_k|U_{ik}|,
\end{align}
where the last equality follows from the fact that each column of $U$ has a unit norm. Let $k_i:= \arg \max_{k} |U_{ik}|$. We show that if $n\delta \le 0.39$, then there does not exist a pair $i\neq j$ with $k_i=k_j$. To do so, let $u = \min_i \max_k |U_{ik}|$. From \eqref{eq:U_ik} we have that
\begin{align}
1-\delta \le \frac1n\sum_{i} \max_k |U_{ik}| \le \frac{n-1}n + \frac{u}{n},
\end{align}
which concludes
\begin{align}\label{eq:bound-on-u}
u\ge 1-n\delta.
\end{align}
Now, by contradiction, assume there exist $i\neq j$ such that $k_i=k_j$. Let $\bm u_i$ and $\bm u_k$ denote the $i$-th and $k$-th columns of $U$. Letting $k=k_i=k_j$, for $\ell\in\{i,j\}$ we can write $\bm u_\ell = u_{\ell k}\bm e_k + \bm r_\ell$ such that $\bm r_\ell^\dag \bm e_{k}=0$. As each column of $U$ has unit norm, we get that
\begin{align}
\norm{\bm r_\ell}^2\le 1-|u_{\ell k}|^2\le 2n\delta - n^2\delta^2, \text{ for }\ell=i,j.
\end{align}
where the last inequality is due to \eqref{eq:bound-on-u}. However, due to orthogonality of $\bm u_i$ and $\bm u_j$ we also have
\begin{align}
0=\bm u_i^\dag \bm u_j = u_{ik} u_{jk} + \bm r_i^\dag \bm r_j
\end{align}
which implies
\begin{align}\label{eq:uu}
|u_{ik} u_{jk}| = |\bm r_i^\dag \bm r_j|\le \norm{\bm r_i}\cdot\norm{\bm r_j}\le 2n\delta -n^2\delta^2.
\end{align}
However, we also have $u^2\le |u_{ik} u_{jk}|$ which combined with \eqref{eq:bound-on-u} and \eqref{eq:uu} gives
\begin{align}
1-n\delta \le 2n\delta -n^2\delta^2.
\end{align}
This is violated for all $n\delta < \frac{3-\sqrt 5}{2} \approx 0.38$. Therefore, we conclude that if $n\delta\le \frac{3-\sqrt 5}{2}$, then all $(k_i)_{i=1}^n$ must be distinct.

Then, we can now define $T$ to be the matrix that has only the following non-zero entries
\begin{align}
T_{ik_i} = \frac{U_{ik_i}}{|U_{i k_i}|}.
\end{align}
Since $k_i$ are all distinct, we get that $T$ is of the form $T=PD$.
In the rest of the proof, we show that $\norm{U - T}$ satisfies the mentioned bound. Note that
\begin{align}
\norm{U - T}^2 = \sum_{ik} |U_{ik} - T_{ik}|^2 = \sum_{i} (1-|U_{ik_i}|^2) + (1-|U_{ik_i}|)^2 \leq 3\sum_{i} 1-|U_{ik_i}| \leq 3n\delta
\end{align}
where the second to last inequality follows from inequalities $(1-x)^2\leq 1-x$ and $1-x^2 \leq 2(1-x)$ that hold for all $0\leq x\leq 1$.
\end{proof}
Going back to the original problem, we note that
\begin{align}
\langle P_{\mathcal G}(\mathbf z) | P_{\mathcal G} (U\mathbf z) \rangle = \frac{1}{12 n} \langle C(\mathbf z)|C(U\mathbf z)\rangle + \frac{1}{2|E|} \langle Q_{\mathcal G}(\mathbf z)|Q_{\mathcal G}(U\mathbf z)\rangle.
\end{align}
Therefore, $\abs{\langle P_{\mathcal G}(\mathbf z) | P_{\mathcal G} (U\mathbf z) \rangle} \geq 1 - \varepsilon$ implies that
\begin{align}\label{eq:both-bounds-eps}
\abs{\frac{1}{6 n} \langle C(\mathbf z)|C(U\mathbf z)\rangle} \geq 1-2\varepsilon,\quad \text{and} \, \frac{1}{|E|} \abs{\langle Q_{\mathcal G}(\mathbf z)|Q_{\mathcal G}(U\mathbf z)\rangle} \geq 1-2\varepsilon.
\end{align}
For any $\varepsilon\le 0.19/n$, from \cref{lem:helper-gapped-cv} we obtain that there exists a diagonal phase matrix $D$ and a permutation matrix $P$ such that
\begin{align}\label{eq:closeness-of-u}
\norm{U - PD} \leq \sqrt{6n\varepsilon}.
\end{align}
Furthermore, we note that
\begin{align}
\frac{1}{2|E|} \langle Q_{\mathcal G_1}(\mathbf z)|Q_{\mathcal G_2}(U\mathbf z)\rangle = \frac{1}{|E|} \mathrm{tr}\left( A_1U^tA_2U \right).
\end{align}
By using the second bound in \eqref{eq:both-bounds-eps} we obtain
\begin{align}\label{eq:auau}
1-2\varepsilon \leq \frac{1}{2|E|} \abs{\mathrm{tr}\left( A_1U^tA_2U \right)}.
\end{align}
Also, letting $\Delta := U - PD$ we have \eqref{eq:closeness-of-u} we have
\begin{align}\label{eq:last-bound-cv}
\begin{split}
1-2\varepsilon &\leq \frac{1}{2|E|} \abs{\mathrm{tr}\left( A_1U^tA_2U \right)}\\
&\leq \frac{1}{2|E|} |\mathrm{tr}\left( A_1 D P^t A_2 P D\right)| + \frac{1}{2|E|} \abs{\mathrm{tr}\left( A_1\Delta^t A_2 PD \right)} + \frac{1}{2|E|}\abs{\mathrm{tr}\left( A_1 DP^t A_2 \Delta \right)} + \frac{1}{2|E|} \abs{\tr(A_1 \Delta^t A_2 \Delta)}\\
&\overset{(i)}{\leq} \frac{1}{2|E|} \abs{\mathrm{tr}\left( A_1 D P^t A_2 P D\right)} + 2\sqrt{6n\varepsilon} + 6n\varepsilon\\
&\overset{(ii)}{\leq} \frac{1}{2|E|} \mathrm{tr}\left( A_1P^t A_2 P\right) + 2\sqrt{6n\varepsilon} + 6n\varepsilon
\end{split}
\end{align}
where $(i)$ is using $\abs{\mathrm{tr}(A_1\Delta^t A_2 PD)} \leq \norm{A_1} \norm{A_2} \norm{\Delta} \norm{PD}$ together with $\norm{\Delta}\leq \sqrt{6n\varepsilon}$ from \eqref{eq:closeness-of-u} and also together with the fact that $\norm{PD} = 1$ and $\norm{A_i}^2 = 2|E|$. Finally, in $(ii)$ we have used the fact that $\abs{\mathrm{tr}(B)}\leq \mathrm{tr}((|B|_{ij})_{ij}) $, or in other words, taking entry-wise absolute value can only increase the absolute value of the trace.\footnote{We are also using $\abs{\mathrm{tr}(AP^tAP)} = \mathrm{tr}(AP^tAP)$.} Lastly, plugging $\varepsilon = \frac{1}{96n^5}$ into \eqref{eq:last-bound-cv} we get
\begin{align}
\frac{1}{2|E|}\abs{\mathrm{tr}\left( A_1 P^t A_2 P \right)} \geq 1 - \frac{1}{n^2}.
\end{align}
We now use a final auxiliary lemma to complete our proof.
\begin{lem}
Let $A$ and $B$ be adjacency matrices over a graph of size $n$, and both having $|E|$ edges. If 
\begin{align}
\frac{1}{2|E|}\tr(AB) \geq 1-\frac{1}{n^2}
\end{align}
then $A = B$.
\end{lem}
\begin{proof}
Note that
\begin{align}
\mathrm{tr}(AB) = \sum_{ij} A_{ij} B_{ij} \leq 2|E|-1
\end{align}
if $A\neq B$. Therefore, if
\begin{align}
\frac{1}{2|E|}\tr(AB) > 1 - \frac{1}{2|E|},
\end{align}
then $A=B$. As $|E|< n^2/2$, we conclude that the condition in the statement of the lemma ensures equality of the graphs.
\end{proof}

This shows that if $\beta > 1-\frac{1}{96 n^5}$, then the two graphs are isomorphic. The contrapositive of this statement concludes that if $G_1 \not\cong G_2$ then the overlap must be smaller than $1-\frac{1}{96n^5}$.
\end{proof}

Next, we establish an upper bound for this problem.

\subsection{$\SZK$ upper bound}

We reduce this problem to the Statistical Difference (\textsc{StatDist}), an $\SZK$-complete problem \cite{sahai2003complete}, and hence, prove containment in $\SZK$.

\begin{theorem} [Formal restatement of the upper bound in \cref{thm:bosonic-PSGI*}]
For every constant $\alpha<1$ and fixed cutoff $r$, it is the case that $(\alpha,\beta;r)$-\textsc{PSGI}$[\mathcal O_n]$
is contained in $\SZK$ whenever
$\beta\geq 1-O(1/n^r)$ (with a sufficiently small hidden constant).
\end{theorem}

\begin{proof}
We show a reduction to the \textsc{StatDist} problem. The input states lie in the Fock subspace with at most $r$ photons, whose dimension is
$d={n+r\choose r}=\Theta(n^r)$. Hence, for $r=O(1)$ we can classically write down these states. We use this fact to produce random vectors in the truncated Hilbert space. We show that if the two rays are close up to a linear optical element, then the random vectors have small total-variation distance; if they are far, the random vectors have far distributions.

Below is a concrete description of how to get the desired phase-invariant sampler from a given core state. Given a core state $\ket{c}$ (with stellar polynomial $P_c$), define sampler $S_c$ that:
\begin{enumerate}
    \item Takes a random unitary $U$ from the Haar measure on $\mathrm{U}(n)$ and an independent uniformly random phase $\theta\in[0,2\pi)$, and calculates $\ket{c_{U,\theta}}:=e^{i\theta}R(U)\ket c$.
    \item Samples a random vector $\ket{\eta}$ from a $\mu$ approximation to the Gaussian distribution $\sim \mathcal N_{\mathbb C}(0, \sigma^2 \mathbb I)$ in the truncated Hilbert space of dimension $d={n+r\choose r}=\Theta(n^r)$.
    \item Outputs $\ket{c_{U,\theta}} + \ket{\eta}$.
\end{enumerate}

In what follows, we use $T(\cdot,\cdot)$ to denote the total variation distance between two random variables.
We use the following fact, which we then use to show that the addition of the Gaussian random noise helps making the two distributions close in YES cases (intuitively, if the orbits are close, then the random Gaussian noise would make them indistinguishable).
\begin{lem}\label{lem:gauss-smooth}
Let $X, Y$ be two random Gaussian vectors of variance $\sigma^2\mathbb I$, and mean vectors $u,v$. We have
\begin{align}
T(X,Y) \le \frac{\norm{u-v}}{2\sigma}.
\end{align}
\end{lem}
\begin{proof}
The proof follows by the fact that the KL-divergence between $X,Y$ is given by
\begin{align}
D_{\mathrm{KL}}(X\|Y) = \frac{\norm{u-v}^2}{2\sigma^2}. 
\end{align}
Then, via Pinsker's inequality we get
\begin{align}
T(X,Y) \le \sqrt{\frac{1}{2}D_{\mathrm{KL}}(X\|Y)} = \frac{\norm{u-v}}{2\sigma}.
\end{align}
\end{proof}
Note that \cref{lem:gauss-smooth} also implies that if $X=Y+e$ for some random $e$ with $\norm{e}\le \delta$ (almost surely) then we have $T(X\|Y)\le \delta/2\sigma$, due to the convexity of the total variation distance. As a result, we have that the total variation distance between the $\mu,\nu$-approximate sampler and the exact sampler is bounded by $(r\nu+\mu)/2\sigma$. Hence, choosing $\nu=\mu/r=o(\sigma^2)$\footnote{As we show below, we have $\sigma = \Theta(1/n^{r/2})$}, we can analyze the problem as if we had sampled from the exact random distributions we desire. The approximate sampling would have only slight effect of $o(1)$ on the total variation distances we calculate.

For normalized vectors, define a distance that measures how far the Gaussian orbits are
\begin{align}
d(c_1,c_2)
:=\inf_{\theta\in[0,2\pi),\,G\in\mathcal O_n}
  \norm{\ket{c_1}-e^{i\theta}G\ket{c_2}}.
\end{align}
We claim that
\begin{lem}\label{lem:szk-yes}
If $d(c_1,c_2)\leq a$, then
\begin{align}
T(S_1, S_2) \leq \frac{a}{2\sigma}.
\end{align}
\end{lem}
We then show a converse
\begin{lem}\label{lem:szk-no}
If $d(c_1,c_2)\geq b$, then
\begin{align}
T(S_1, S_2) \geq \frac23,
\end{align}
assuming $\sigma = \Theta\left(\frac{b}{n^{r/2}}\right)$.
\end{lem}
The identity
\begin{align}
d(c_1,c_2)^2
=2-2\sup_{G\in\mathcal O_n}\abs{\bra{c_1}G\ket{c_2}}
\end{align}
shows that the YES promise gives $a=\sqrt{2(1-\beta)}$, while the NO promise gives $b=\sqrt{2(1-\alpha)}$. Choose
$\sigma=c_0b/\sqrt d$, with the absolute constant $c_0$ small enough that
the NO-case Gaussian tail below has probability at least $5/6$. This gives us the desired result.
\begin{proof}[Proof of \cref{lem:szk-yes}]
Let $X_i$ be the output of $S_i$. Condition on the choices $(U,\theta)$ in
the first sampler. By hypothesis, there are $W\in\mathrm U(n)$ and
$\theta_0\in[0,2\pi)$ such that
\[
    \norm{\ket{c_1}-e^{i\theta_0}R(W)\ket{c_2}}\le a.
\]
Couple the first sampler's pair $(U,\theta)$ with the second sampler's pair
$(UW,\theta+\theta_0)$. Their mean vectors obey
\[
    \norm{e^{i\theta}R(U)\ket{c_1}
    -e^{i(\theta+\theta_0)}R(UW)\ket{c_2}}\le a.
\]
By \cref{lem:gauss-smooth}, the corresponding conditional Gaussian
distributions have total-variation distance at most $a/(2\sigma)$. Mixing over
$U$ and $\theta$ can only decrease this distance. Haar invariance makes $UW$
Haar random, and translation invariance makes $\theta+\theta_0$ uniform, so
the two marginals are exactly $S_1$ and $S_2$. This proves the claim (with the
stated, looser constant).
\end{proof}

\begin{proof}[Proof of \cref{lem:szk-no}]
Let the linear optical orbit of the $i$-th input be
\begin{align}
\mathcal V_i:=\{e^{i\theta}R(U)\ket{c_i}:
U\in\mathrm U(n),\ \theta\in[0,2\pi)\}.
\end{align}
It is the case that
\begin{align}
\inf_{\ket x, \ket y: \ket x\in \mathcal V_1, \ket{y} \in \mathcal V_2} \norm{x-y} \ge b.
\end{align}
Now, define the region of distance $b/2$ from $\mathcal V_1$ by
\begin{align}
\mathcal Z := \left\{ \ket{z}: \inf_{\ket{x}\in\mathcal V_1}\norm{x-z} \le \frac{b}2\right\}.
\end{align}
Since a Gaussian's norm with covariance $\sigma \mathbb I_d$ is highly concentrated around $\sigma\sqrt{d}$, we get that with the choice of $\sigma = \Theta(\frac{b}{\sqrt{n^{r}}})$ we are guaranteed
\begin{align}
\Pr(\norm{\eta} \le b/2) = p,
\end{align}
for some $p=\Theta(1)$, where $\ket\eta \sim \mathcal N_{\mathbb C}(0, \sigma^2 \mathbb I)$. As a result, we get that
\begin{align}
\begin{split}
&\Pr(X_1\in\mathcal Z) \geq \Pr(\norm{\eta}\le b/2) = p,\\
&\Pr(X_2\in\mathcal Z) \le \Pr(\norm{\eta}\ge b/2) = 1-p.
\end{split}
\end{align}
As a result, we obtain that
\begin{align}
T(S_1,S_2) \ge 2p-1.
\end{align}
Setting $p=5/6$ gives the desired result.
\end{proof}
Lastly, the Verifier must use finite-precision representations of both $U$ and
$\theta$. With polynomially many bits of precision, their meshes can be chosen
so that the induced error in each mean vector is $o(\sigma)$; by
\cref{lem:gauss-smooth}, the total-variation distances then change by only
$o(1)$.
\end{proof}

Finally, we note that an $\mathsf{NP}$ upper bound for any $\beta-\alpha\ge \frac{1}{2^{\mathsf{poly}(n)}}$ is straightforward, as the prover can send the unitary $U$ that corresponds to the passive map that approximately maps one state into another.

\printbibliography

\end{document}